\documentclass[5p,times,sort&compress]{elsarticle}

\usepackage[cmex10]{amsmath}
\usepackage{amssymb,amsthm}
\usepackage{algorithmic,algorithm}
\usepackage{mdwmath,mdwtab}
\usepackage{eqparbox}
\usepackage[bold,full]{complexity}
\usepackage[font={footnotesize}]{caption}
\usepackage[font={footnotesize}]{subcaption}
\usepackage{psfrag}
\usepackage{textcomp}
\usepackage{array}
\usepackage{color}
\usepackage{pdfpages}
\usepackage{graphicx}
\newtheorem{theorem}{Theorem}
\usepackage{microtype}

\begin{document}

\begin{frontmatter}
\title{Stochastic Buffer-Aided Relay-Assisted MEC\\ in Time-Slotted Systems}
\author{Javad Hajipour}
\address{Department of Computer Science, University of Tabriz, Tabriz 5166616471, Iran.\\hajipour@tabrizu.ac.ir}
\begin{abstract}
Mobile Edge Computing (MEC) has attracted significant research efforts in the recent years. However, these works consider mostly the computation resources located at the cloud centers and wireless access nodes, ignoring the possibility of utilizing server-empowered relays to improve the performance. In this paper, we study stochastic relay-assisted MEC in systems with discrete transmission time-line and block fading wireless channels. In order to clearly identify and inspect the fundamental affecting factors, we investigate the building block of this architecture, namely a hierarchical network consisting of a source, a buffer-and-server-aided relay and another higher-level computing node. We provide a framework to take into account the effects of the fading channels, the task arrival dynamics as well as the queuing delays in both the transmission and computation buffers, which facilitates the derivation of the expression for the Average Response Time (ART). Based on that and the system average power consumption in each slot, we introduce the concept of Average Response Energy (ARE) as a novel metric to capture the energy efficiency in MEC while considering the stochastic nature of the system parameters. Accordingly, we propose two offloading schemes with their respective problem formulations, namely the Minimum ART (MART) and the Minimum ARE (MARE) schemes, to optimize the transmission power and task assignment probability while keeping the system queues stable. We demonstrate the difference of the formulated problems with the relevant problem in a recent work, analyze the properties of the problems and noting them, we propose effective solution methods. Using extensive simulations, we validate the presented analysis and show the effectiveness of the proposed schemes in comparison with various baseline methods adapting existing approaches.\looseness=-1
\end{abstract}
\end{frontmatter}

\vspace*{-.2cm}
\section{Introduction}\label{sec:intro}
The advent of the Internet of Things (IoT) and the massive device connections in wireless networks have resulted in the proliferation of new applications utilized by these devices such as augmented reality, face recognition, and health-related applications~\cite{JR:class_opt_prob,JR:class_fog_app}. Consequently, the network providers have encountered new challenges to address the task computation requests from these devices in a timely manner. In this regard, Mobile Edge Computing (MEC) has emerged as a new approach to reduce the Average Response Time (ART) for the requested tasks by the IoT Devices (IoTDs) and other resource-constrained mobile stations~\cite{JR:het_edge_op_plat,JR:review_edge_ref,JR:fog_stat_art,JR:edge_feat_virtual}. In MEC, different components of the wireless networks are empowered with computing resources, which enable them to process the requested tasks by their own servers. This alleviates the need to send the tasks to the servers of the clouds located in remote centers, and reduces the delay for the applications. Therefore, there has been remarkable research and investigations on MEC, in the recent years. These works will be reviewed briefly in the next subsection and after that, we will present the motivation and contributions of this paper.

\subsection{Background and Literature Review}\label{subsec:rel_works}
Many studies have investigated the potentials and challenges in MEC systems while considering possible scenarios and the related constraints~\cite{JR:surv_offl_model_mec}. In particular, hierarchical computation architectures were studied in~\cite{JR:hier_auction,JR:hier_cloudlet}, where a three-level hierarchy of computation facilities with two time scale model of frames and slots within each frame was considered in~\cite{JR:hier_auction}. The authors proposed an auction-based policy for allocating the computation and communication resources in a way to maximize the service provider's profit in each frame and minimize the delay of the users in each slot. In~\cite{JR:hier_cloudlet}, the ART of the users' requests was minimized by optimizing the cloudlet assignment and computation resource allocation.~\cite{JR:joint_lb_offl_veh} considered MEC in vehicular networks and maximized system utility for joint load balancing and offloading.~\cite{JR:dynamic_congestion} aimed at minimizing the average completion time of applications in a multi-user scenario, while considering the dependencies between the tasks, the limitations on communication and computation resources as well as the congestion level on these resources. The authors first formulated a static problem and then adapted it to dynamic systems, and proposed an effective method to make decision for task offloading.~\cite{JR:stoch_resal_pw_mec} studied task offloading in a wireless powered MEC system with a single user that receives wireless energy and computation assistance from a server-empowered access point. The authors formulated a stochastic optimization problem to maximize the average task execution rate and, using the Lyapunov optimization technique, they proposed an online algorithm for resource management which takes into account the length of task queue, the battery level and the channel states in each time slot. In~\cite{JR:reg_intel}, a regional intelligent system for vehicular networks was presented using MEC. The authors considered different types of tasks and aimed at optimizing the resource allocation to minimize the delay. They formulated the problem as a Markov decision process and used deep Q-learning to address that.
~\cite{JR:het_elastic} studied  an edge computing system that depending on the workload of the tasks, requests to tenant or release computing instances from the cloud resources and accordingly, decides about data migration and replica placement.~\cite{JR:opt_fair_aw} studied a multi-user multi-server MEC system, where each user's task can be partitioned into several segments to be computed locally or assigned to edge servers, and a fairness-aware method was proposed to minimize the maximum delay for task computations in the system.
~\cite{JR:del_opt_IRS,JR:IRS_d2d_coop} studied MEC systems in the presence of intelligent reflecting surface (IRS) which is able to improve the signal propagation environment by configuring its reflecting elements in a way to enhance the signal reception quality at the receiver.~\cite{JR:del_opt_IRS} presented a combination of time-division multiple access (TDMA) and non-orthogonal multiple access (NOMA) as the data transmission scheme for the task offloading by the two users and minimized the sum delay of the users.~\cite{JR:IRS_d2d_coop} studied an IRS-assisted device-to-device (D2D) cooperative computing system where several nodes help a source node in computing the portions of a task. The authors aimed at designing the computation task assignment, transmit power and bandwidth allocation, and phase beamforming of the IRS to minimize the computing delay. In~\cite{JR:svc_plc_req_rout}, a cloud radio access network with multiple edge servers and a centralized cloud was investigated, where decision for service placement on the edge servers and cloud are made over longer timescale whereas user request routing to the servers and cloud is conducted within a shorter timescale. Efficient algorithms were proposed for joint optimization of service placement and request routing in order to obtain maximum system utility under the constraints related to the storage and computation resources.~\cite{JR:GEESE,JR:trajec_uav} studied edge computation provisioning through unmanned autonomous vehicles (UAVs) and investigated the advantages as well as the challenges of this approach. Many other works have discussed exploiting artificial intelligence and machine learning techniques to enhance the decision making and resource management in MEC systems~\cite{JR:task_mig_reinforc19,JR:intel_ml_smart_city21,JR:dron_learn_video21,JR:autonomous_deep_learn21}.

Energy efficiency is also one of the important challenges in Information and Communication Technology (ICT) management and hence, many works have investigated that in MEC systems.~\cite{JR:ee_mec_offl17} investigated energy-efficient offloading scheme in TDMA and orthogonal frequency division multiple access (OFDMA) MEC systems in the presence of the users with the same delay constraint. The authors investigated the cases with infinite and finite computing capacity at the edge and proposed low-complexity algorithms to minimize the weighted sum energy consumption of the users. In~\cite{JR:MEC_assoc_multi_task18}, a multi-task scenario with sequential dependencies of the tasks was taken into account and an algorithm was designed for the user association and computation offloading that minimizes the energy consumption.~\cite{JR:eng_aw_green_pow19} considered a MEC system enhanced by the energy Internet (EI), where the edge nodes have the ability to transfer the green energy among themselves, and formulated a problem to minimize the brown energy consumption. Taking the task loads and energy transfer attenuation into account, they proposed a heuristic algorithm with near-optimal performance for deciding about the computation resource migration as well as task allocation and energy scheduling.~\cite{JR:joint_offl_ee19} studied energy-efficient computation offloading in a time-slotted multi-user system by defining the computation efficiency as the number of the computed bits divided by the corresponding energy consumption. The authors aimed at maximizing the computation efficiency while considering the constraints on the total energy of the devices, and proposed an iterative algorithm to obtain the optimal values for the processor speed and transmit power of the nodes as well as the slot portion allocation to different nodes.~\cite{JR:scalable_e_d19} studied energy and delay trade-off in a multi-user MEC system with a multi-antenna full-duplex (FD) base station (BS) that sends information in the downlink and receives computation task requests through the uplink transmissions of the users. The authors proposed a beamforming and power allocation scheme to minimize the weighted sum of the offloading energy and latency while considering the constraints on the maximum delay, maximum power and minimum signal to interference-plus-noise ratio (SINR). In~\cite{JR:en_eff_user_asgn_pwalloc20}, energy-efficient task offloading was investigated in a MEC system with multiple users and multiple edge servers. The authors formulated an optimization problem and proposed a solution method to maximize the ratio of total transmit rate of all users to the total power consumption by jointly optimizing user-edge server assignments and user transmit powers.~\cite{JR:opt_en_wirl_pow} investigated a single-user system with multi-antenna wireless energy transmitter (ET) and minimized energy consumption of ET over a finite time horizon by jointly optimizing the energy allocation and task offloading.~\cite{JR:eng_aw_fog_cluster21} investigated an edge computing system with energy harvesting and exchanging capability at the nodes. They used a clustering scheme based on the energy transfer attenuation ratio to group the nodes into clusters where the nodes inside the cluster send their computation requests to the computing node of that cluster. They also proposed a heuristic method to decide about task schedule and energy transfer within/between the clusters.\looseness=-1

\subsection{Motivation and Contributions}\label{subsec:mot_contr}
The above-mentioned works mostly considered the hierarchy of the computing nodes as an architecture consisting of the servers at the wireless access nodes and the severs at the cloud computing centers. Even though they investigated many sophisticated scenarios, they did not take into account the various possibilities that can be brought by the relay nodes. But we note that, in the past decade, many works demonstrated the advantages of using relays, especially the buffer-aided relays, for improving the performance of the wireless networks~\cite{JR:bufaid_adaptive,JR:buf_fixedmixed,JR:BufImp_without_link,JR:EERA_barel,JR:CRS}. In particular,~\cite{JR:bufaid_adaptive} studied adaptive link selection in a three-node network comprising a source, a buffer-aided relay and a destination, and demonstrated significant throughput gains due to the use of buffer at the relay. A similar system was considered in~\cite{JR:buf_fixedmixed} with fixed-rate and mixed-rate transmission scenarios, where the transmission rate at the source and relay are fixed or the relay adapts its transmission rate according to the channel state information, when available. 
For both cases, improvement in the throughput was observed. In~\cite{JR:BufImp_without_link}, it was shown that buffer-aided relaying improves not only the throughput, but also the end-to-end delay. In fact, when considering the whole picture and taking into account the queuing delay at the source, it is revealed that the increase in the system throughput leads to lower latency since the data arrival at the source until the reception at the destination. The advantages of buffer-aided relays were also demonstrated in the multi-user and multi-relay scenarios, e.g.,~\cite{JR:EERA_barel} investigated energy-efficiency gain in multi-user scenarios and~\cite{JR:CRS} studied outage improvement in relay selection schemes in multi-relay networks.\looseness=-1

Noting the aforementioned, there is a major research gap on using the relay nodes to enhance the MEC systems. Recently, in~\cite{JR:sbaramec}, a Multi-hop MEC (MMEC) architecture was proposed where it was shown that employing a computation capability at a buffer-aided relay and using a stochastic offloading scheme can improve the ART, significantly. But the results were derived based on the assumption that fixed transmission rates are used on the links; therefore, it was only applicable for the scenarios that the wireless channels do not experience deep fading and the transmitting nodes are able to adapt their transmit power according to the variations of the channel conditions. However, in practice, there are many scenarios that the wireless channels experience deep fading which does not permit successful transmission at some periods. Moreover,~\cite{JR:sbaramec} assumed continuous time domain where the nodes can start transmission/computation of the tasks/results whenever there is one in the corresponding transmission/computation buffers. However, many systems have discrete transmission time-line and the nodes are only allowed to transmit at the beginning of system-level-defined time intervals. Thus, more investigations are required to analyze such systems, which are the subject of the current paper.\looseness=-1 

Specifically, in contrast to the existing research works, we study stochastic buffer-aided relay-assisted MEC in the time-slotted systems with block fading channels, for the first time to the author's best knowledge. As for many new architectures and protocols, it is required to study the building block in the early stages of developing the relay-assisted MEC, to get clear insights about the effect of the fundamental parameters and the possible outcomes. Moreover, it is important to provide a framework which will give directions for investigating the next developments and the more sophisticated scenarios with different constraints. Hence, similar to~\cite{JR:bufaid_adaptive,JR:buf_fixedmixed,JR:BufImp_without_link,JR:sbaramec}, we consider a three-node network as the building block of hierarchical time-slotted MMEC systems. It is composed of a source, a buffer-aided server-enabled relay, and another server-empowered node with a level higher than the level of the relay in the hierarchy, where the relay makes a random decision to assign a received task to its own server or to the server of the next node. However, the presented framework and discussions can be extended to MMEC systems with more than two hops and/or more than one source and/or relay. We believe that the combination of MEC and buffer-aided relaying will play remarkable role in provisioning ubiquitous access to computation resources, in the emerging wireless networks densely populated with IoTDs. In particular, buffer-and-server-aided relay-assisted MEC possesses promising potentials for both ad hoc and infrastructure-based wireless networks to exploit the ever-growing computation capabilities of the user equipments and/or relay nodes scattered in these networks and to pave the way for high performance applications. 
Therefore, we are motivated to initiate the investigations in this area and build a solid basis for further research.
In summary, the main contributions of this work are as follows:
\begin{itemize}
\item We provide the system model and establish a framework to capture the factors affecting the ART in a time-slotted system with block fading channels and fixed transmit power on the links. In particular, we discuss the queuing model of the transmission buffers and derive their service probabilities.
\item We analyze the ART in the system and compare its form with the counterpart in the continuous-time fixed-transmit-rate systems. Based on that, we propose the Minimum ART (MART) scheme with the problem formulation to find the optimal values for the transmit power on the links and the probability of task assignment to the relay, such that the lowest ART is achieved while keeping the system queues stable. We discuss the properties of the formulated problem and prove that the optimal value of the transmit power is equal to its maximum value. Then, we prove that the problem is convex with respect to the task assignment probability and therefore, it can be solved using a low-complexity one-dimensional search method.
\item We demonstrate that at high values of the transmit power, increasing the power does not have much effect on the ART. We derive the Slot Average Power (SAP) as a function of the transmit power on the links and the task assignment probability. Based on that, we introduce a novel notion of energy efficiency as the product of the SAP and the ART in the system, yielding the concept of Average Response Energy (ARE). Using that, we propose the Minimum ARE (MARE) scheme and formulate a new optimization problem to obtain a low ART and at the same time use the system power efficiently. We explain the properties of the new objective function and provide directions to reduce the complexity of the solution strategies. Moreover, in the special case of Rayleigh fading, we demonstrate that the problem is convex with respect to the transmit power and based on that, we propose effective algorithms to find the optimal values of the transmit power and task assignment probability.
\item We conduct extensive simulations to verify the presented analysis and also demonstrate the effectiveness of the proposed schemes in comparison with the related baseline methods. Moreover, we evaluate the performance of the proposed schemes in different system settings and provide valuable insights on the possible trends of the outcomes.
\end{itemize}

It is worth emphasizing that the proposed schemes differ from the previous works as we consider different system model employing a buffer-and-server-aided wireless relay node between the source node and a higher-level computing node in a time-slotted MEC system. By taking into account the dynamics of task arrivals, channel variations over time slots, stochastic offloading decisions, and their effects on the transmission and computation queues, the equation derived for the ART in the time-slotted systems is different from that in the existing works. Moreover, the ARE is a new objective function that we define in order to design an energy-efficient resource allocation method while taking into consideration the stochastic nature of the aforesaid parameters. Consequently, the formulated problems are novel and require new analysis. Furthermore, the essential insights provided by the presented analysis and the numerical results can be exploited in the future works for designing sophisticated solutions for more complex and multi-hop systems. In fact, this work is a starting point to exploit the wireless relays for expanding the hierarchy of MEC architecture in time-slotted systems; therefore, it has the potential to inaugurate new research trends and investigations on the architectures that incorporate the buffer-and-server-aided relays into the system models studied previously in the recent years. 

The remainder of this paper is organized as follows. In Section~\ref{sec:sysmodel}, the system model is described. Section~\ref{sec:MART_MARE}, presents the proposed offloading schemes and the corresponding problem formulations. We analyze their properties and present the proposed solution methods. In Section~\ref{sec:results}, numerical results are provided and Section~\ref{sec:conclusion} presents the conclusion.

\vspace*{-.2cm}
\section{System Model}\label{sec:sysmodel}
We consider a three-node network, as shown in Fig.~\ref{fig:sysmodel}, comprising a Source Node (SN), a Relay Node (RN), and another node called the Higher-level Node (HN). The SN may represent an IoTD or a wireless mobile station that does not have a computing server. The HN and the RN may be respectively the BS and an installed wireless relay in a cellular network, or they may represent any stations in an ad hoc network that have computation capability and cooperate to process the tasks demanded by the SN. It is assumed that time is partitioned into the units of slot; in every slot, an application in the SN generates a task with the arrival probability $a$ and independent from the other tasks. The SN stores the tasks in a transmission buffer to send to the RN whenever possible. On the other hand, the RN exploits a stochastic scheme to decide about processing or offloading the received tasks. In particular, when a task is received by the RN, it is assigned to the server of the RN with probability $\rho$, and is offloaded to the HN with probability $1-\rho$. The results of the tasks processed at the HN are transmitted back to the RN and together with the results of the tasks processed at the RN, they all are sent to the SN gradually.

\begin{figure}[t]
\centering
\psfrag{u}[scale=.18]{$q^t_{sr}$}
\psfrag{v}[scale=.18]{$q^t_{rs}$}
\psfrag{w}[scale=.18]{$q^c_r$}
\psfrag{x}[scale=.18]{$q^t_{rh}$}
\psfrag{y}[scale=.18]{$q^t_{hr}$}
\psfrag{z}[scale=.18]{$q^c_h$}
\includegraphics[scale=.55]{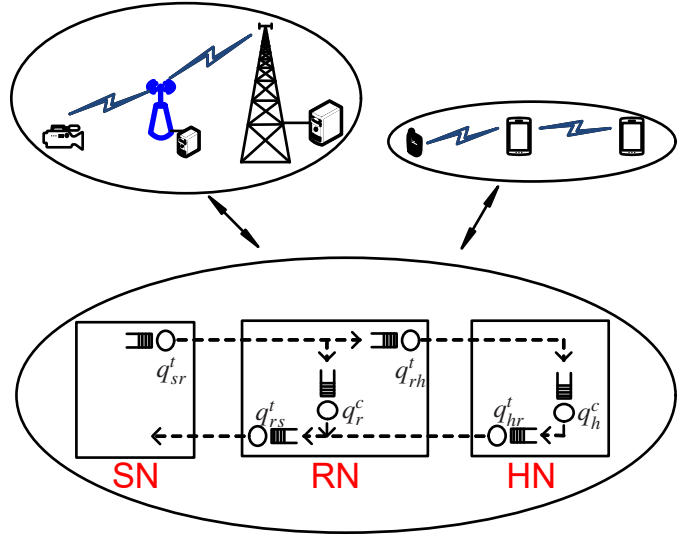}
\caption{System model.}
\vspace*{-.3cm}
\label{fig:sysmodel}
\end{figure}

There are separate dedicated orthogonal channels for transmissions on the SN-RN, RN-HN, HN-RN, and RN-SN
links, and a direct SN-HN link does not exist. The RN has a computation buffer for storing the tasks to be computed by its own server. Moreover, it has a transmission buffer to store the tasks to be forwarded to the HN and a transmission buffer to store the results to be sent to the SN. Similarly, the HN is equipped with two buffers, one for storing the tasks to be computed and another one for storing the results to be sent to the RN. At the beginning of each time slot, if there is a task/result in a transmission (computation) buffer and its serving channel (computation server) is available, the task/result is removed from the buffer and transmitted (computed) completely by the end of the slot. We assume that the propagation delay between the nodes is negligible. The response time for a task is defined as the time duration since the the task is generated until the result is received at the SN. Hence, this time duration consists of the delays at the computation and transmission queues which in turn are affected by the service probabilities of the queues. Based on the discussions above, the arrival and service processes of the buffers are Bernoulli processes and thus, the number of time slots between the arrivals and between the services of each buffer has geometric distribution. For such queuing systems with arrival probability $x$ and service probability $y$, the queue is stable if $x<y$~\cite{book:stochasticopt}; in that case, the probability of a departure in a slot is also $x$ and the average waiting time (in units of slot) is~\cite{book:intro_que_telecom}:
\begin{equation}\label{eq:Wxy}
W(x,y)=\frac{1-x}{y-x},
\end{equation}
which includes the time spent in the queue and in the server of the queue.
We note that the computation servers of the RN and the HN may be busy due to processing their own tasks and hence, they may not be available in every time slot to process a task of the SN. The probability that in a time slot, the server of the RN (HN) can allocate its computation resources to process a task of the SN is denoted by $q^c_r$ ($q^c_h$). On the other hand, sending the tasks/results from a transmission buffer is only possible in the slots that the signal-to-noise-ratio (SNR) at the receiver side of the serving link is higher than a threshold, $\gamma$. 
We assume that the channel conditions are ergodic and stationary processes; they stay fixed over a time slot but change from one slot to the next. Let $g_{sr}$, $g_{rs}$, $g_{rh}$, and $g_{hr}$ denote the squared channel gains, respectively, on the SN-RN, RN-SN, RN-HN, and HN-RN links in a general time slot. We use $\mathcal{F}_{g_{sr}}(g)$ to denote the Cumulative Distribution Function (CDF) of $g_{sr}$ which indicates the probability of having $g_{sr}\leq g$. $\mathcal{F}_{g_{rs}}(g)$, $\mathcal{F}_{g_{rh}}(g)$, and $\mathcal{F}_{g_{hr}}(g)$ are defined in a similar way. As an early work to combine MEC and buffer-aided relaying, we assume that all the nodes use identical and constant transmit power $P$. The reason for considering identical transmit power at all the nodes is that it allows to have similar transmission capability for all the nodes and focus only on the impact of the single power variable and the task assignment probability; however, the problem formulations can be extended in future works to include the cases with different transmit powers at the nodes as well as other constraints.

Let $N_0$ denote the noise power at the receiver side of the links. The SNR on the SN-RN link is equal to $\frac{Pg_{sr}}{N_0}$; based on that and the aforementioned, the probability of having the SNR on the SN-RN link higher than $\gamma$, which indicates the probability of channel availability for successful task transmissions from the SN-RN transmission buffer, equals 
\begin{equation}\label{eq:qt_sr}
q^t_{sr}(P)=1-\mathcal{F}_{g_{sr}}(\frac{\gamma}{P/N_0}).
\end{equation}
$q^t_{rs}(P)$, $q^t_{rh}(P)$, and $q^t_{hr}(P)$ can be defined in a similar way. In the rest of the paper, we may refer to $q^t_{sr}(P)$, $q^t_{rs}(P)$, $q^t_{rh}(P)$, and $q^t_{hr}(P)$ as the channel service probabilities and, for brevity, we may omit the argument $P$ in some equations. We assume that channel service probabilities are increasing functions of $P$, because as the transmit power increases, the probability of satisfying the threshold SNR, $\gamma$, at the receiver is increased. 

Moreover, to be concrete, we will specialize the derived results for the case of Rayleigh fading channels. Noting the characteristics of Rayleigh distribution, we have $\mathcal{F}_{g_{sr}}(g)=1-e^{-\frac{g}{\overline{g}_{sr}}}$, where $\overline{g}_{sr}$ is the mean value of $g_{sr}$. Thus, using~\eqref{eq:qt_sr} yields $q^t_{sr}=e^{-\frac{\Gamma_{sr}}{P}}$, where $\Gamma_{sr}=\frac{\gamma}{\overline{g}_{sr}/N_0}$. Similarly, for the RN-SN, RN-HN, and HN-RN links respectively, we have $q^t_{rs}=e^{-\frac{\Gamma_{rs}}{P}}$, $q^t_{rh}=e^{-\frac{\Gamma_{rh}}{P}}$, and $q^t_{hr}=e^{-\frac{\Gamma_{hr}}{P}}$,
where $\Gamma_{rs}$, $\Gamma_{rh}$, and $\Gamma_{hr}$ are defined similar to $\Gamma_{sr}$.

\section{Problem Formulations and Solutions; The MART and MARE Schemes}\label{sec:MART_MARE}
Based on the discussions in the previous section, in order to ensure queue stability, the service probability should be higher than the arrival probability for each queue. 
In that case, since the arrived tasks at the RN are randomly kept with probability $\rho$ or offloaded to the HN with probability $1-\rho$, the task arrival probability for the computation buffer at the RN will be $\rho a$. For the RN-HN and HN-RN transmission queues as well as the HN computation queue, the task/result arrival probability will be $(1-\rho)a$. Since the results of the tasks processed at the RN and HN enter the RN-SN transmission queue, the arrival probability of the results at the RN-SN transmission queue is $a$.\looseness=-1

For the tasks processed at the RN, the ART consists of the delays in the SN-RN transmission queue, the RN computation queue, and the RN-SN transmission queue. Hence, the ART of these tasks is obtained as 
\begin{equation}\label{eq:Tr}
T_r(P,\rho) = W(a,q^t_{sr})+W(\rho a,q^c_r)+W(a,q^t_{rs})
\end{equation}
Similarly, the tasks processed at the HN have the ART 
\begin{align}\label{eq:Th}
T_h(P,\rho)=&W(a,q^t_{sr})+W((1-\rho)a,q^t_{rh})+W((1-\rho)a,q^c_h)&\nonumber\\
+&W((1-\rho)a,q^t_{hr})+W(a,q^t_{rs}).&
\vspace*{-.2cm}
\end{align}
Hence, the ART of the tasks in the system is 
\begin{equation}\label{eq:T}
T(P,\rho) = \rho T_r(P,\rho) + (1-\rho)T_h(P,\rho).
\end{equation}

Based on the aforesaid, we present two problem statements in the following with different objective functions; then, we discuss their properties and propose solution methods.

\subsection{The MART Scheme}
We note that the ART equation derived in the previous section is different from that in~\cite{JR:sbaramec}. In particular, due to the time-slotted nature of the current system and the queuing models, the numerators of the constituting terms of the ART are nonlinear (second-order polynomial) functions of $\rho$. Moreover, here, we also take into account the effect of transmit power on the ART, and we desire to obtain the optimal transmit power, $P^*$, and the optimal probability of task assignment to the RN, $\rho^*$, that lead to the minimum ART in the system. Specifically, we aim at addressing the following problem which is referred to as the MART problem:

\vspace*{-.2cm}
\begin{subequations}
\label{eq:MART_problem}
\begin{align}
& \displaystyle \min_{P,\rho} \ \ \ T(P,\rho) \label{eq:MART_objective},& \\
\text{s.t.}
& \quad C1: a <q^t_{sr}(P) ,				 \label{eq:t_sr_stability_con}&\\
& \quad C2: a <q^t_{rs}(P) ,				 \label{eq:t_rs_stability_con}&\\
& \quad C3: (1-\rho)a<q^t_{rh}(P) , \label{eq:t_rh_stability_con}&\\
& \quad C4: (1-\rho)a<q^t_{hr}(P) , \label{eq:t_hr_stability_con}&\\
& \quad C5: \rho a<q^c_r , 			 \label{eq:c_r_stability_con}&\\
& \quad C6: (1-\rho)a<q^c_h , 	 \label{eq:c_h_stability_con}&\\
& \quad C7: P\in [0,\hat{P}] ,					 \label{eq:pmax_con}&\\
&\quad  C8: \rho \in [0,1], 		 \label{eq:rho_domain_con}&
\end{align}
\end{subequations}
where, $C1-C6$ state the stability conditions for the buffers of the system; $C7$ indicates that the maximum power a node can use is $\hat{P}$ and $C8$ states the domain of $\rho$.

In general, problem~\eqref{eq:MART_problem} is not a convex optimization problem not least because constraints~\eqref{eq:t_sr_stability_con}-\eqref{eq:t_hr_stability_con} cannot be stated in the form $f_i(P,\rho)\leq0$ with $f_i(P,\rho)$s being convex functions, due to the fact that $q^t_{sr}(P)$, $q^t_{rs}(P)$, $q^t_{rh}(P)$, and $q^t_{hr}(P)$ may not be concave functions of $P$ (which is the case in particular for Rayleigh fading). Moreover, since the objective function~\eqref{eq:MART_objective} is the sum of several nonlinear functions of $P$ and $\rho$, getting its Hessian matrix is expensive and complicated; therefore, it is difficult to verify whether the objective function is convex or not, even in the case of Rayleigh fading. Thus, in the following, we analyze the properties of problem~\eqref{eq:MART_problem} to obtain clues for finding $P^*$ and $\rho^*$. We note that constraints~\eqref{eq:t_sr_stability_con} and~\eqref{eq:t_rs_stability_con} do not affect $\rho$ and if they are satisfied for a given $P\in [0,\hat{P}]$, then, the remaining constraints will determine the existence of feasible $\rho$.\looseness=-1

\begin{theorem}\label{thm:MART_feasible_rho}
For a given $P$ that satisfies~\eqref{eq:t_sr_stability_con},~\eqref{eq:t_rs_stability_con} and~\eqref{eq:pmax_con}, problem~\eqref{eq:MART_problem} has feasible values for $\rho$ if and only if 
\begin{equation}
\label{eq:rho_feasibility_cond}
\frac{q^c_r+\min(q^t_{rh}(P),q^c_h,q^t_{hr}(P))}{a}>1, 
\end{equation}
and the set of feasible values for $\rho$ is $\mathcal{A}=\mathcal{S}\cup(\max(0,1-\frac{\min(q^t_{rh}(P),q^c_{h},q^t_{hr}(P))}{a}),\min(1,\frac{q^c_r}{a}))$, where
\vspace*{-.2cm}\begin{eqnarray}\label{eq:boundaries}
 \mathcal{S} =  \left\{ \begin{array}{lll} 
\{1,0\} &  \mathrm{if} \ \ q^c_r> a , \min(q^t_{rh}(P),q^c_{h},q^t_{hr}(P))> a, \\  
\{1\} 	&  \mathrm{if} \ \ q^c_r> a  , \min(q^t_{rh}(P),q^c_{h},q^t_{hr}(P))= a, \\ 
\{0\}	  &  \mathrm{if} \ \ q^c_r= a  , \min(q^t_{rh}(P),q^c_{h},q^t_{hr}(P))> a, \\
\{\} 		&  \mathrm{otherwise}. \end{array} \right.
\end{eqnarray}
\end{theorem}
\begin{proof}
This can be proved easily similar to the first theorem in~\cite{JR:sbaramec}.
\end{proof}

In the following theorem, $I(x)$ is unit step function defined to be one if $x>0$ and zero otherwise.
\begin{theorem}\label{thm:MART_feasible_P}
If problem~\eqref{eq:MART_problem} is feasible, then $P \in (\check{P},\hat{P}]$, where $\check{P}=\max(\check{P}_1,\check{P}_2)<\hat{P}$, $\check{P}_1=\max({q^t_{sr}}^{-1}(a),{q^t_{rs}}^{-1}(a))$, $\check{P}_2=I(a-q^c_r)\max({q^t_{rh}}^{-1}(a-q^c_r),{q^t_{hr}}^{-1}(a-q^c_r))$, where superscript ``$-1$" for a functions indicates the inverse of that function.
In the special case of Rayleigh fading channels, $\check{P}_1=-\frac{\max(\Gamma_{sr},\Gamma_{rs})}{\ln a}, \check{P}_2=-I(a-q^c_r)\frac{\max(\Gamma_{rh},\Gamma_{hr})}{\ln(a-q^c_r)}$.
\end{theorem}

\begin{proof}
If the problem is feasible, there exists $P$ and $\rho$ that satisfy the constraints. Based on~\eqref{eq:t_sr_stability_con} and~\eqref{eq:t_rs_stability_con}, we have the inequality $\min(q^t_{sr}(P),q^t_{rs}(P))>a$ which is equivalent to $P>\check{P}_1=\max({q^t_{sr}}^{-1}(a),{q^t_{rs}}^{-1}(a))$. Moreover, based on~\eqref{eq:rho_feasibility_cond}, we have $\min(q^t_{rh}(P),q^c_h,q^t_{hr}(P))>a-q^c_r$; this necessitates the inequality $\min(q^t_{rh}(P),q^t_{hr}(P))>a-q^c_r$ which holds for any $P>0$ if $a-q^c_r\leq 0$, and for $P>\max({q^t_{rh}}^{-1}(a-q^c_r),{q^t_{hr}}^{-1}(a-q^c_r))$ if $a-q^c_r>0$; these two cases can be concisely stated as $P>\check{P}_2=I(a-q^c_r)\max({q^t_{rh}}^{-1}(a-q^c_r),{q^t_{hr}}^{-1}(a-q^c_r))$. Hence, $P>\check{P}=\max(\check{P}_1,\check{P}_2)$. Noting~\eqref{eq:pmax_con}, we conclude that $\check{P}<P\leq\hat{P}$.
For Rayleigh fading channels, by substituting the corresponding equations for the channel service probabilities, we get $\check{P}_1=-\frac{\max(\Gamma_{sr},\Gamma_{rs})}{\ln a}$, $\check{P}_2=-I(a-q^c_r)\frac{\max(\Gamma_{rh},\Gamma_{hr})}{\ln(a-q^c_r)}$. 
\end{proof}

For conciseness, in the following theorem we use $q^t_{rhr}(P)$ to refer to $\min(q^t_{rh}(P),q^t_{hr}(P))$.
\begin{theorem}\label{thm:A_size}
By increasing $P$ from $\check{P}$ towards $\hat{P}$, one of the following cases will occur for the size of $\mathcal{A}$:
\begin{itemize}
	\item $H1$: It does not change, if $q^t_{rhr}(\check{P})\geq \min(q^c_h,a)$.
	\item $H2$: It continuously increases, if $q^t_{rhr}(\hat{P})\leq \min(q^c_h,a)$.
	\item $H3$: It continuously increases as $P$ increases from $\check{P}$ up to $\tilde{P}=\max({q^t_{rh}}^{-1}(\min(q^c_h,a)),{q^t_{hr}}^{-1}(\min(q^c_h,a)))$ and then it does not change, if $q^t_{rhr}(\check{P})<\min(q^c_h,a)$ and $q^t_{rhr}(\hat{P})>\min(q^c_h,a)$. In the case of Rayleigh fading, $\tilde{P}=-\frac{\max(\Gamma_{rh},\Gamma_{hr})}{\ln(\min(q^c_h,a))}$.
\end{itemize}
\end{theorem}

\begin{proof}
According to Theorem~\ref{thm:MART_feasible_rho}, the upper boundary of $\mathcal{A}$ is independent from $P$ and the lower boundary depends on $\min(q^c_h,q^t_{rhr})$. Based on Theorem~\ref{thm:MART_feasible_P}, $q^t_{rhr}$ is minimum at $\check{P}$ and maximum at $\hat{P}$, as $q^t_{rh}$ and $q^t_{hr}$ are increasing functions of $P$. Hence, we have the following cases:

\begin{itemize}
\item If $q^t_{rhr}(\check{P})\geq \min(q^c_h,a)$, either $a>q^c_h$ and $q^t_{rhr}(\check{P})\geq q^c_h$ and hence, the lower boundary of $\mathcal{A}$ is determined by $1-\frac{q^c_h}{a}$, or $q^c_h\geq a$ and $q^t_{rhr}(\check{P})\geq a$ and hence, the lower boundary of $\mathcal{A}$ is 0. In either case, the lower boundary is independent from $q^t_{rh}$, $q^t_{hr}$ and $P$. This remains true as $P$ increases and hence, case $H1$ holds.

\item If $q^t_{rhr}(\hat{P})\leq q^c_h$ and $q^t_{rhr}(\hat{P})\leq a$, the lower boundary of $\mathcal{A}$ is determined by $1-\frac{q^t_{rhr}(P)}{a}$ for any $P\in(\check{P},\hat{P}]$. Therefore, as $P$ increases in that interval, $q^t_{rhr}(P)$ increases and the lower boundary of $\mathcal{A}$ decreases and hence, the size of $\mathcal{A}$ increases. Therefore, case $H2$ holds.

\item If none of the conditions above holds, then $q^t_{rhr}(\check{P})< q^c_h$ and $q^t_{rhr}(\check{P})< a$ but $q^t_{rhr}(\hat{P})> q^c_h$ or $q^t_{rhr}(\hat{P})>a$. In this situation, at low values of $P\in(\check{P},\hat{P}]$, $q^t_{rhr}(P)$ determines the lower boundary of $\mathcal{A}$. Hence, as $P$ increases, $q^t_{rhr}(P)$ increases and the lower boundary decreases and the size of $\mathcal{A}$ increases; but this continues up to the point that $q^t_{rhr}(P)$ equals $q^c_h$ or $a$, whichever happens first, i.e., up to $\tilde{P}=\max({q^t_{rh}}^{-1}(\min(q^c_h,a)),{q^t_{hr}}^{-1}(\min(q^c_h,a)))$, because after that point, either $q^c_h$ determines the lower boundary or the lower boundary is 0. This proves case $H3$. For Rayleigh fading, exploiting the equations for channel service probabilities, we have $\tilde{P}=-\frac{\max(\Gamma_{rh},\Gamma_{hr})}{\ln(\min(q^c_h,a))}$.\looseness=-1
\end{itemize}
\end{proof}

\vspace*{-.2cm}
\begin{theorem}\label{thm:MART_optimal_P}
$P^*$ is equal to $\hat{P}$.
\end{theorem}
\begin{proof}
Let $\hat{\mathcal{S}}$ denote the feasible set of $\rho$ at $P=\hat{P}$. For any $\rho\in\hat{\mathcal{S}}$, decreasing the $P$ will decrease $q^t_{sr}$, $q^t_{rs}$, $q^t_{rh}$, $q^t_{hr}$ and hence, according to Theorems~\ref{thm:MART_feasible_rho}-\ref{thm:A_size}, the problem may remain feasible or not; but according to~\eqref{eq:T}, the objective function will certainly increase. This concludes the theorem.
\end{proof}

The above-mentioned theorem can be explained as follows. As $P$ increases, the right-hand sides of the constraints~\eqref{eq:t_rh_stability_con} and~\eqref{eq:t_hr_stability_con} increase, which allow for a larger range of feasible $\rho$ and therefore, the possibility for a $\rho$ with lower ART increases. Moreover, the higher the $P$, the lower the delay at the transmission queues.\looseness=-1

Based on the insights gained by the discussions above, problem~\eqref{eq:MART_problem} can be transformed to a single-variable convex optimization problem. For that, the channel service probabilities should be calculated using $P^*=\hat{P}$, and substituted in the problem. 

\begin{theorem}\label{thm:MART_convexity_rho}
For a given $P$, in particular for $P^*$, problem~\eqref{eq:MART_problem} is a convex optimization problem with respect to $\rho$.\looseness=-1
\end{theorem}
\begin{proof}
For a given $P$, channel service probabilities are constant and the constraints of problem~\eqref{eq:MART_problem} are linear functions with respect to $\rho$. Thus, we only need to prove that the objective function is convex or, equivalently, show that it has a positive second derivative~\cite{book:nonlinear_bertsekas}.~\eqref{eq:Wxy}-\eqref{eq:T} indicate that the objective function is composed of the functions of the forms $f_0=\frac{1-a}{q_0-a}$, $f_1(\rho)=\frac{\rho-\rho^2a}{q_1-\rho a}$ and $f_2(\rho)=\frac{(1-\rho)-(1-\rho)^2a}{q_2-(1-\rho)a}$, where $q_0$, $q_1$, and $q_2$ represent the symbols for transmission/computation service probabilities and $f_0$ is a function that does not depend on $\rho$. The second derivatives of $f_1(\rho)$ and $f_2(\rho)$ are as follows:
\vspace*{-.2cm}
\begin{equation}\label{eq:f_second_der}
f_1^{''}(\rho) = \frac{2aq_1(1-q_1)}{(q_1-\rho a)^3} \ ,\ \ \ f_2^{''}(\rho) = \frac{2aq_2(1-q_2)}{(q_2-(1-\rho)a)^3}.
\vspace*{-.2cm}
\end{equation}
It is clear that the values of $a$, $q_1$, and $q_2$ are in the interval $(0,1]$ and, according to the constraints, $(q_1-\rho a)>0$, $(q_2-(1-\rho)a>0)$ for feasible $\rho$. Therefore, $f_1^{''}(\rho)>0$ and $f_2^{''}(\rho)>0$ hold and hence, $f_1(\rho)$ and $f_2(\rho)$ are convex functions. Consequently, it is inferred that for a given $P$ and in particular for $P^*$, the objective function~\eqref{eq:MART_objective} is convex with respect to $\rho$. 
\end{proof}

According to Theorem~\ref{thm:MART_convexity_rho}, after substituting the channel service probabilities (at $P^*=\hat{P}$) in the problem, the resulted problem will be a single-variable convex optimization problem with respect to $\rho$. Therefore, a one-dimensional search algorithm can be utilized to easily find the $\rho^*$. For example, Golden section method~\cite{book:nonlinear_bertsekas} is an efficient method in this regard, with the computational complexity of $O(\log(1/\delta))$, where $\delta$ is the desired relative error bound. This method can be easily implemented in the RN or the HN as they both have computation capability.

\begin{figure}[b]
\centering
\includegraphics[scale=.35]{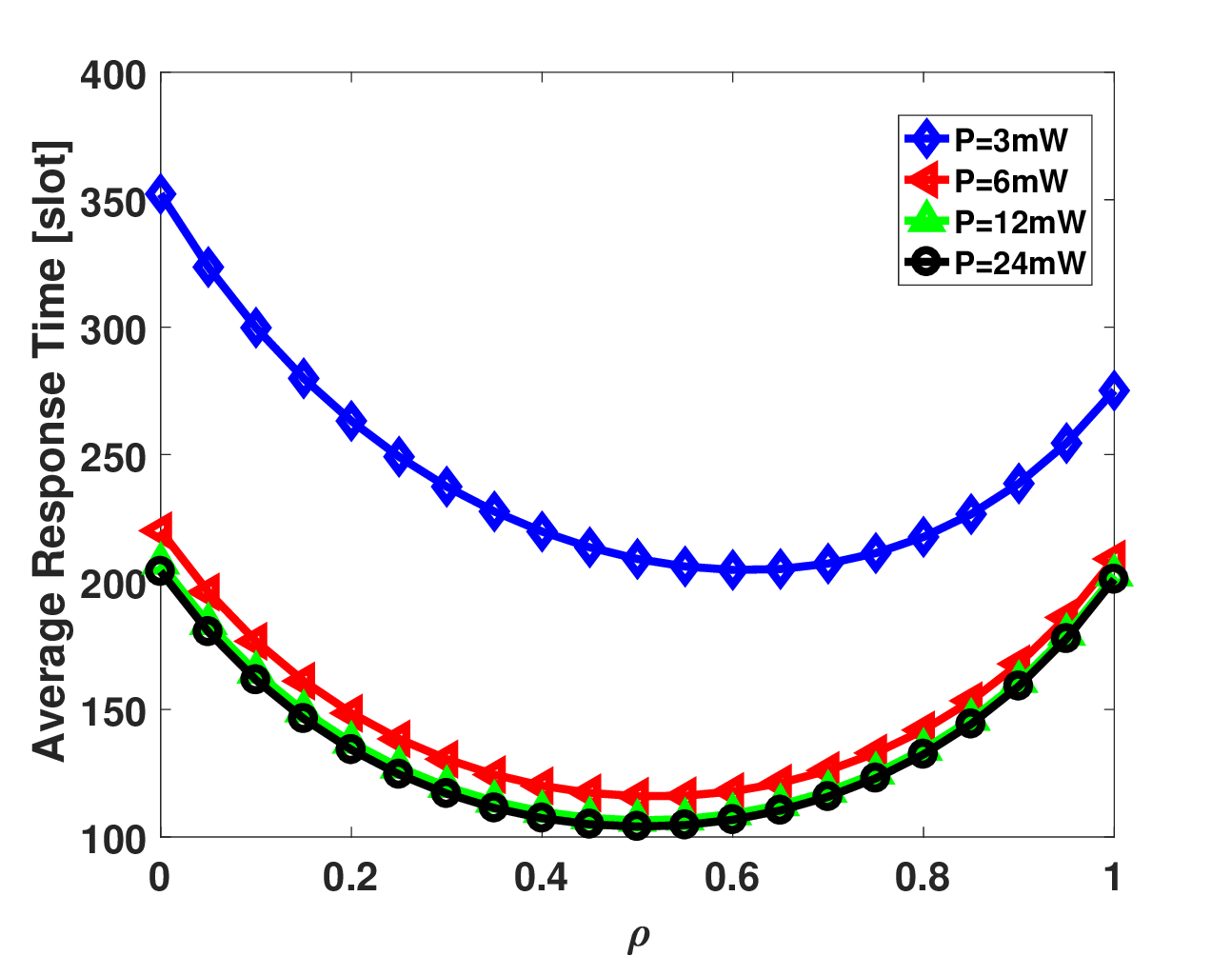}
\caption{Effect of different values of $P$ on the ART in a system with $a=0.01$, $q^c_r=q^c_h=0.015$, and Rayleigh fading with $\Gamma_{sr}=\Gamma_{rs}=\Gamma_{rh}=\Gamma_{hr}=0.01$.}
\vspace*{-.5cm}
\label{fig:efficiency}
\end{figure}

\subsection{The MARE Scheme}
In many scenarios, we are interested to have a low ART, not necessarily the minimum one, and at the same time use the system energy efficiently. To clarify that, Fig.~\ref{fig:efficiency} illustrates the ART with respect to $\rho$ for different values of $P$, in a system with $a=0.01$, $q^c_r=q^c_h=0.015$, and Rayleigh fading with $\Gamma_{sr}=\Gamma_{rs}=\Gamma_{rh}=\Gamma_{hr}=0.01$. It is observed that for any $P$, the ART is a convex function with respect to $\rho$ and a global minimum exists. This confirms the discussions in the previous section and indicates that using the maximum $P$ and the optimal value of $\rho$, the ART of the system can be minimized. However, we also observe that for any $\rho$, when $P$ is increased from $3$ mW to $6$ mW, there is a significant improvement in the ART; but when $P$ is increased more, the ART does not change much. In other words, the Rate of Decrease (RoD) in the ART is very smaller than the Rate of Increase (RoI) in the energy consumption.

Noting the aforementioned, we are motivated to strike a balance between the energy consumption in the system and the response time of the tasks. We know that the energy consumption on a link is proportional to the average power consumption of that link in each slot. On the other hand, the average power consumption of a link is proportional to the ratio of the time slots that there is a transmission on that link. Hence, the average power consumptions in each slot on the SN-RN, RN-SN, RN-HN, and HN-RN links are respectively equal to $aP$, $aP$, $(1-\rho)aP$, and $(1-\rho)aP$. We define the SAP in the system as the sum of the average power consumptions of all the links in each slot, i.e.,
\begin{equation}\label{eq:P_s}
\bar{P}_s(P,\rho)=2(2-\rho)aP. 
\end{equation}

In order to use the system energy efficiently, we take into account both the RoI in the energy consumption and the RoD in the ART and define a novel objective function as the product of the SAP and the ART in the system. This product can be interpreted as the average energy consumption during the response time of the tasks. Thus, we refer to it as the ARE, denoted by
\begin{equation}\label{eq:ARE}
E(P,\rho)=\bar{P}_s(P,\rho)T(P,\rho),
\end{equation}
and propose the MARE offloading scheme with the following optimization problem:

\vspace*{-.2cm}
\begin{subequations}
\label{eq:MARE_problem}
\begin{align}
& \displaystyle \min_{P,\rho} \ \ \ E(P,\rho)  \label{eq:MARE_objective},& \\
\text{s.t.}
& \quad C1-C8. &
\end{align}
\end{subequations}

Similar to the MART problem, in general, the MARE problem is not a convex optimization problem. It is clear that Theorems~\ref{thm:MART_feasible_rho}-\ref{thm:A_size} hold also for the MARE problem, since it has the same constraints as the MART problem. However, due to the multiplication of $(2-\rho)P$ and $T(P,\rho)$ terms in the objective function, neither of Theorems~\ref{thm:MART_optimal_P} and~\ref{thm:MART_convexity_rho} are applicable for the MARE problem. Therefore, another numerical method should be used to find $P^*$ and $\rho^*$.
Even, running an exhaustive search with a desired precision over two dimensions of $P$ and $\rho$ is not very expensive, computationally, for two reasons. First, it is only over two dimensions. Second, the problem is only needed to be solved offline and once, i.e., at the beginning of the operation, not in every time slot; after that, all the nodes will use $P^*$ for their transmissions and the RN will work based on $\rho^*$ whenever it receives a task from the SN. However, reducing the needed steps for solving a problem is always desirable. In the following, we discuss the properties of the objective function with respect to $\rho$ in general and with respect to $P$ in the especial case of Rayleigh fading; these provide directions to reduce the number of iterations for searching the solution of the MARE problem.

\textbf{Remark 1:} For a given $P$, the constraints of the MARE problem are linear with respect to $\rho$. The objective function $E(P,\rho)$ is the sum of the linear decreasing functions, namely $(2-\rho)aPW(a,q^t_{sr})$, $(2-\rho)aPW(a,q^t_{rs})$, and the rational functions having the forms $f_3=aP(2-\rho)\rho\frac{1-\rho a}{q_3-\rho a}$ and $f_4=aP(2-\rho)(1-\rho)\frac{1-(1-\rho) a}{q_4-(1-\rho) a}$.
The derivatives of $f_3$ and $f_4$ with respect to $\rho$ are obtained as
\begin{subequations}
\label{eq:f3_f4_der}
\begin{align}
f_3^{'}&=2aP(1-\rho)\left[\frac{1-\rho a}{q_3-\rho a} + \frac{\rho a(1-q_3)}{(q_3-\rho a)^2}\right] \label{eq:f3_der},& \\
f_4^{'}&=aP\left[\frac{(2\rho-3)(1-(1-\rho)a)}{q_4-(1-\rho)a}\right.& \nonumber \\
&\left.+(2-\rho)(1-\rho)a\frac{q_4-1}{(q_4-(1-\rho)a)^2}\right].\label{eq:f4_der}&
\end{align}
\end{subequations}
\begin{algorithm}[!t]
\caption{$K$-Partition Search Over $\rho$}
\label{alg:k-partition}
\begin{algorithmic}[1]
\STATE set $\varphi=\frac{3-\sqrt{5}}{2}$, $k=1$; $P$, $K$ and $\delta$ as intended.
\STATE calculate $q^t_{sr}$, $q^t_{rs}$, $q^t_{rh}$ and $q^t_{hr}$ using $P$.
\STATE set $\rho_l=\max(0,1-\frac{\min(q^t_{rh},q^c_{h},q^t_{hr})}{a})$, $\rho_u=\min(1,\frac{q^c_r}{a})$
\STATE \textbf{if} $\min(q^t_{rh},q^c_{h},q^t_{hr})\leq a$, $\rho_l=\rho_l+\delta$.
\STATE \textbf{if} $q^c_r\leq a$, $\rho_u=\rho_u-\delta$.
\STATE $\rho^*=\rho_l$.
\STATE \textbf{while} $k\leq K$
\STATE \hspace{4mm} $\rho_1=\rho_l+\frac{\rho_u-\rho_l}{K}(k-1)$, $\rho_4=\rho_l+\frac{\rho_u-\rho_l}{K}k$
\STATE \hspace{4mm} \textbf{while} $\rho_4-\rho_1>\delta$
\STATE \hspace{8mm} $\rho_2=\rho_1+\varphi(\rho_4-\rho_1)$, $\rho_3=\rho_4-\varphi(\rho_4-\rho_1)$.
\STATE \hspace{8mm} \textbf{if} $E(P,\rho_2)<E(P,\rho_3)$, $\rho_4=\rho_3$; \textbf{else} $\rho_1=\rho_2$.
\STATE \hspace{4mm} \textbf{end while}
\STATE \hspace{4mm} \textbf{if $E(P,\frac{\rho_1+\rho_4}{2})<E(P,\rho^*)$}, $\rho^*=\frac{\rho_1+\rho_4}{2}$.
\STATE \hspace{4mm} $k=k+1$.
\STATE \textbf{end while}
\STATE return $\rho^*$.
\end{algorithmic}
\vspace{-.1cm}
\end{algorithm}
\begin{algorithm}[!t]
\caption{MARE Algorithm in General Cases}
\label{alg:MARE_general}
\begin{algorithmic}[1]
\STATE set $K$ and $\delta$ as intended.
\STATE initialize $P^*=P_x=\check{P}+\delta$, with $\check{P}$ defined in Theorem~\ref{thm:MART_feasible_P}.
\STATE use Algorithm~\ref{alg:k-partition} with $P=P_x$; store the output in $\rho^*$.
\STATE \textbf{while} $P_x\leq\hat{P}$
\STATE \hspace{4mm} use Algorithm~\ref{alg:k-partition} with $P=P_x$; store the output in $\rho_x$.
\STATE \hspace{4mm} \textbf{if} $E(P_x,\rho_x)<E(P^*,\rho^*)$, \textbf{then} $P^*=P_x$, $\rho^*=\rho_x$.
\STATE \hspace{4mm} $P_x=P_x+\delta$.
\STATE \textbf{end while}
\STATE return $P^*$ and $\rho^*$.
\end{algorithmic}
\vspace{-.1cm}
\end{algorithm}
\begin{algorithm}[!t]
\caption{MARE Algorithm in Rayleigh Fading}
\label{alg:MARE_rayleigh}
\begin{algorithmic}[1]
\STATE set $\varphi=\frac{3-\sqrt{5}}{2}$, $K$ and $\delta$ as intended. 
\STATE initialize $P_1=\check{P}+\delta$ and $P_4=\hat{P}$, with $\check{P}$ defined in Theorem~\ref{thm:MART_feasible_P}.
\STATE \textbf{while} $P_4-P_1>\delta$
\STATE \hspace{4mm} $P_2=P_1+\varphi(P_4-P_1)$, $P_3=P_4-\varphi(P_4-P_1)$.
\STATE \hspace{4mm} use Algorithm~\ref{alg:k-partition} with $P=P_2$; store the output in $\rho_1$.
\STATE \hspace{4mm} use Algorithm~\ref{alg:k-partition} with $P=P_3$; store the output in $\rho_2$.
\STATE \hspace{4mm} \textbf{if} $E(P_2,\rho_1)<E(P_3,\rho_2)$ 
\STATE \hspace{8mm} $P_4=P_3$
\STATE \hspace{4mm} \textbf{else} 
\STATE \hspace{8mm} $P_1=P_2$
\STATE \hspace{4mm} \textbf{end if} 
\STATE \textbf{end while}
\STATE set $P^*=\frac{P_1+P_4}{2}$; use Algorithm~\ref{alg:k-partition} with $P=P^*$ and store the output in $\rho^*$.
\STATE return $P^*$ and $\rho^*$.
\end{algorithmic}
\vspace{-.1cm}
\end{algorithm}
Since $\rho\leq 1$, $q_3\leq 1$, $q_4\leq 1$, $\rho a<q_3$ and $(1-\rho)a<q_4$, we notice that $f_3^{'}\geq 0$ and $f_4^{'}\leq 0$. Thus, for a given $P$, $E(P,\rho)$ is the sum of non-decreasing rational functions and non-increasing rational/linear functions with respect to $\rho$. Hence, we expect $E(P,\rho)$ to be either increasing or decreasing or convex or concave function of $\rho$, depending on the effect of its decreasing and increasing terms. Consequently, for a given $P$, we expect $E(P,\rho)$ to have its minimum at one of the boundaries of $\mathcal{A}$ or to have at most a single local minimum point inside $\mathcal{A}$ which can be found using an efficient one-dimensional search, like Golden section method, over $\rho$. We have verified this through function plots with many different parameter settings. However, it cannot be proved analytically. Therefore, for high reliability, the domain of $\rho$ can be partitioned into several sub-domains for each of which, a local optimal point can be found, and the minimum among all of them can be chosen as the near-optimal point. For clarification, this is presented in Algorithm~\ref{alg:k-partition} where $K$ is the intended number of partitions over the domain of $\rho$\footnote{Note that lines 9-12 of Algorithm~\ref{alg:k-partition} implement the iterations of the Golden section method for each sub-domain of $\rho$. Hence, Algorithm~\ref{alg:k-partition} can also be used to solve the MART problem, by setting $K=1$, $P=P^*$, and using $T(P,\rho)$ instead of $E(P,\rho)$.}.

To summarize, Algorithm~\ref{alg:MARE_general} shows the whole procedure for solving the MARE problem, which is a combination of exhaustive search and Golden section method. The exhaustive search is done over the domain of $P$ and at each step, Algorithm~\ref{alg:k-partition} is utilized to obtain the best $\rho$ at the given $P$. Therefore, the computational complexity of Algorithm~\ref{alg:MARE_general} is $O(\frac{K}{\delta}\log(1/\delta))$.

\begin{figure*}[!t]
  \centering{
\begin{subfigure}[t]{\textwidth}
    \centering
    \includegraphics[width=.8\textwidth]{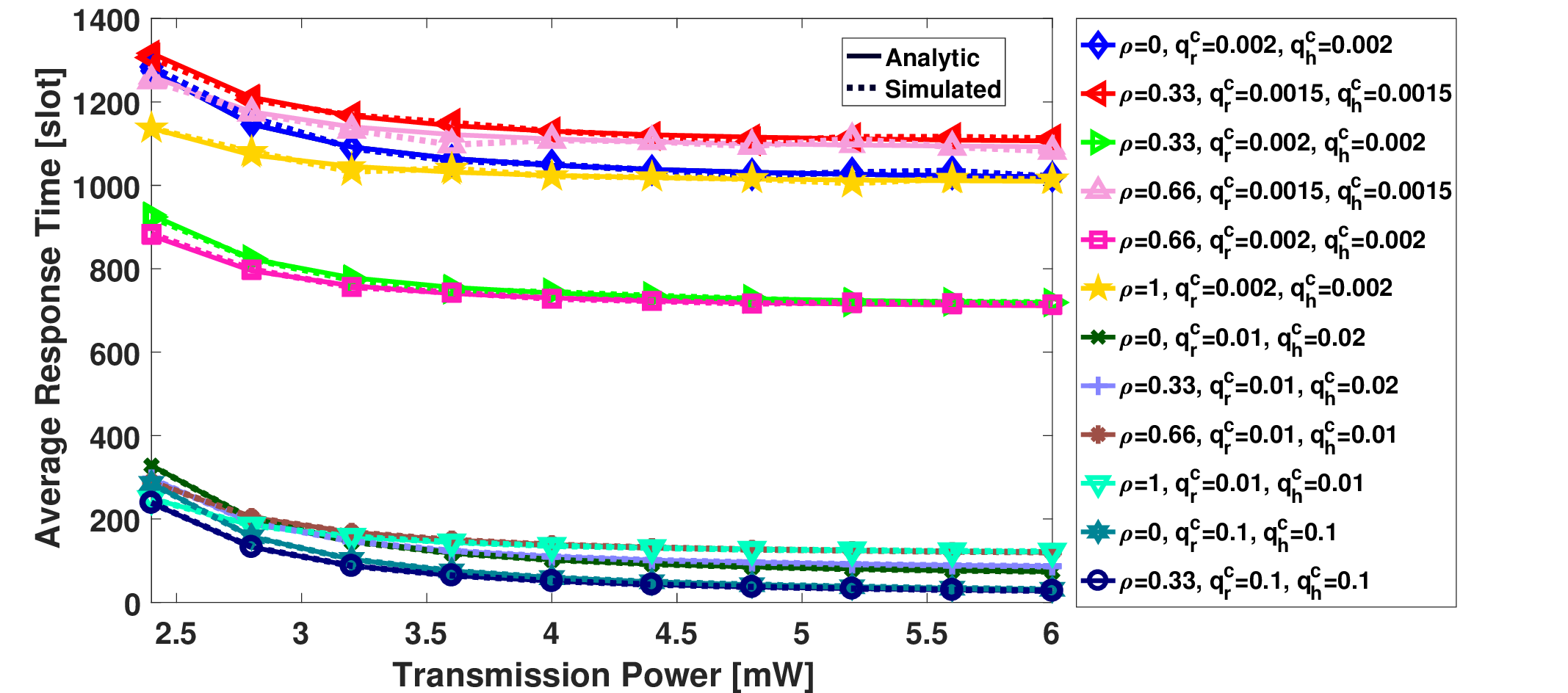}
    \caption{}
    \label{fig:ART_vs_P}
\end{subfigure}\\
\begin{subfigure}[t]{0.48\textwidth}
    \centering
    \includegraphics[width=\textwidth]{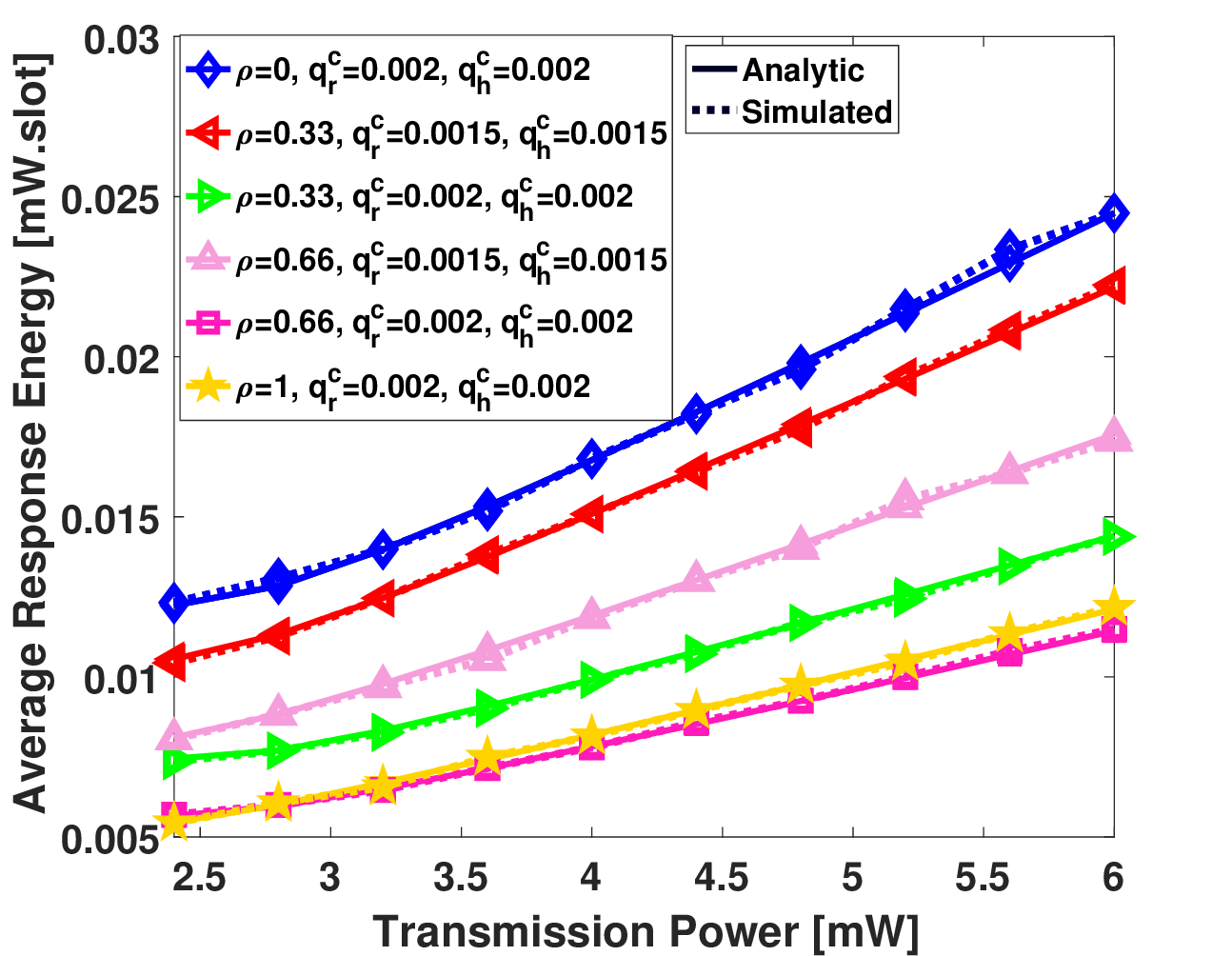}
     \caption{}
    \label{fig:ARE_vs_P1}
\end{subfigure}
\begin{subfigure}[t]{0.49\textwidth}
    \centering
    \includegraphics[width=\textwidth]{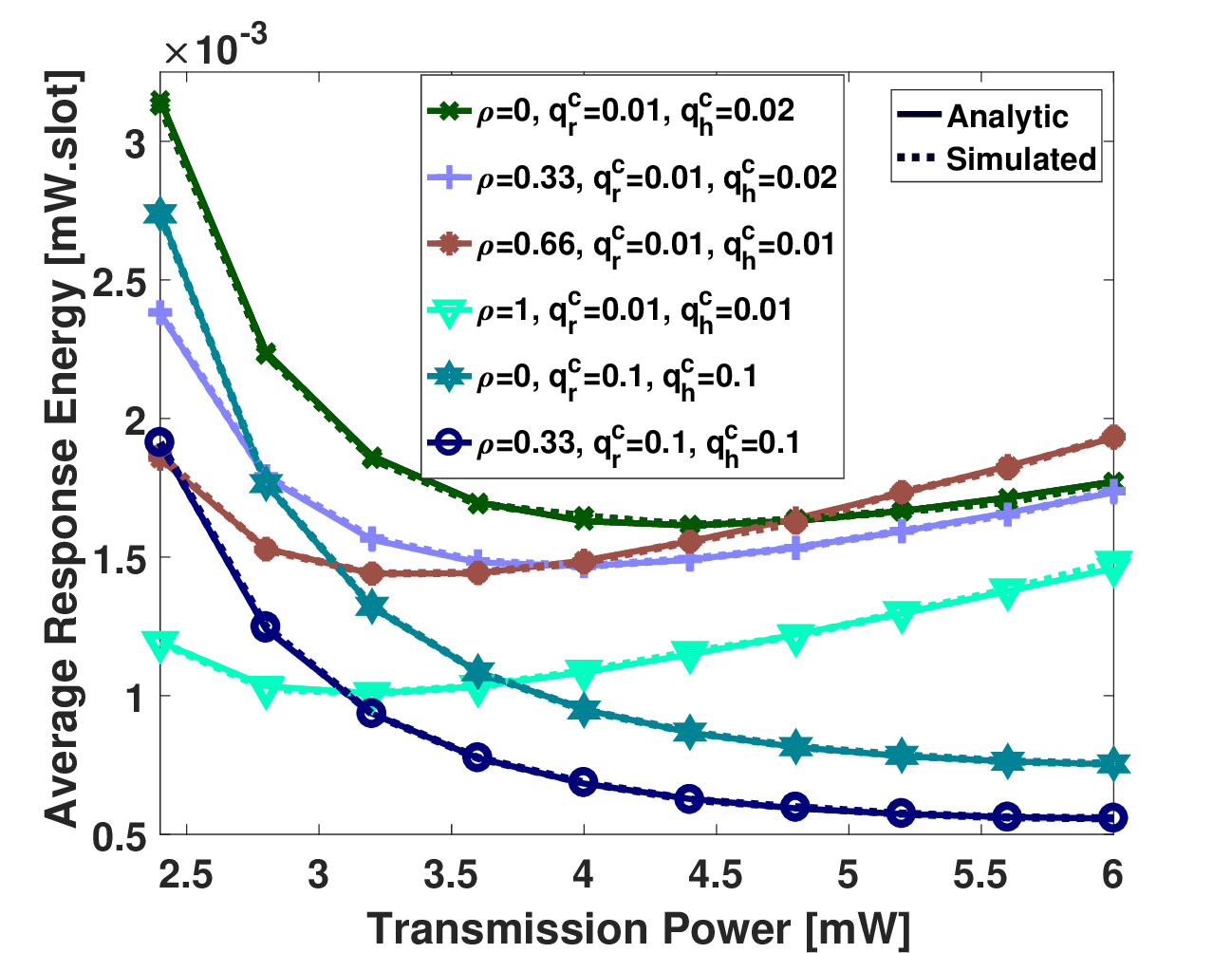}
     \caption{}
    \label{fig:ARE_vs_P2}
\end{subfigure}%
}
\caption{Analytic and simulation results for (a) ART (b, c) ARE, vs. $P$; $a=0.001$.} 
 \label{fig:ART_ARE_vs_P}
 \end{figure*}

\begin{theorem}\label{thm:convexity}
In the case of Rayleigh fading, for a given $\rho$, problem~\eqref{eq:MARE_problem} is a convex optimization problem with respect to $P$.\looseness=-1
\end{theorem}
\begin{proof}
Let substitute the channel service probabilities in the constraints~\eqref{eq:t_sr_stability_con}-\eqref{eq:t_hr_stability_con} with their equations in terms of $P$ in the case of Rayleigh fading. Then, it is observed that the constraints of the problem have the nonlinear form $c_1a<e^{\frac{-\Gamma}{P}}$, where $c_1$ is either 1 or $(1-\rho)$, and $\Gamma$ is one of $\Gamma_{sr}$, $\Gamma_{rh}$, $\Gamma_{hr}$, and $\Gamma_{sr}$ which are positive constants. By taking the logarithm of the two sides of the constraints, they can be transformed into the equivalent linear forms $P<-\frac{\Gamma}{\ln c_1a}$. Thus, it suffices to show that the objective function $E(P,\rho)$ is convex with respect to $P$. According to~\eqref{eq:Wxy}-\eqref{eq:T} and~\eqref{eq:P_s}-\eqref{eq:ARE}, $E(P,\rho)$ is the sum of the linear functions $(2-\rho)\rho aW(\rho a,q^c_r)P$ and $(2-\rho)(1-\rho)aW((1-\rho)a,q^c_h)P$, and the nonlinear functions having the form 
\vspace*{-.2cm}
\begin{equation}\label{eq:f_6}
f_6=(2-\rho)c_1aP\frac{1-c_1a}{e^{\frac{-\Gamma}{P}}-c_1a}. 
\vspace*{-.2cm}
\end{equation}
The second derivative of $f_6$ with respect to $P$ is 
\begin{equation}\label{eq:f_6_second_der}
f_6^{''}=(2-\rho)c_1a(1-c_1a)\frac{\Gamma^2e^{-\frac{\Gamma}{P}}(e^{-\frac{\Gamma}{P}}+c_1a)}{P^3(e^{-\frac{\Gamma}{P}}-c_1a)^3},
\end{equation}
Based on the constraints, we have $c_1a<e^{-\frac{\Gamma}{P}}$ for feasible solutions. Hence, $f_6^{''}>0$ holds, which indicates that $f_6$ and consequently, the objective function~\eqref{eq:MARE_objective} are convex functions of $P$.
\end{proof}

Based on the aforementioned, we propose Algorithm~\ref{alg:MARE_rayleigh} to solve problem~\eqref{eq:MARE_problem} in the case of Rayleigh fading. It searches over the domain of $P$ using the Golden section method and at each $P$, it uses Algorithm~\ref{alg:k-partition} to obtain a near-optimal $\rho$. Therefore, it has the computational complexity of $O(K(\log(1/\delta))^2)$. It is clear that at each iteration, $P_1$ and $P_4$ get closer to each other and therefore, the algorithm will terminate finally. As the results in the next section indicate, the proposed algorithm is highly effective and its performance is close to optimum. 

\vspace*{-.2cm}
\section{Numerical Results and Discussions}\label{sec:results}

In this section, we verify the presented analysis and also evaluate the performance of the proposed schemes. We consider Rayleigh fading channel on all the links, with $\Gamma_{sr}=\Gamma_{rs}=\Gamma_{rh}=\Gamma_{hr}=0.01$ corresponding to $\gamma=10$, $N_0=10^{-3}$ and $\overline{g}_{sr}=\overline{g}_{rs}=\overline{g}_{rh}=\overline{g}_{hr}=1$. The reason for considering similar qualities for all the links is to have equal channel service probabilities on all the links; this enables us to accurately investigate the effect of other deciding factors on the system performance and observe the result of interplay between transmission and computation queues. Hence, we consider several settings for other parameters, as presented in the sequel. 

\begin{figure*}[!t]
  \centering{
\begin{subfigure}[t]{0.49\textwidth}
    \centering		
    \includegraphics[width=\textwidth]{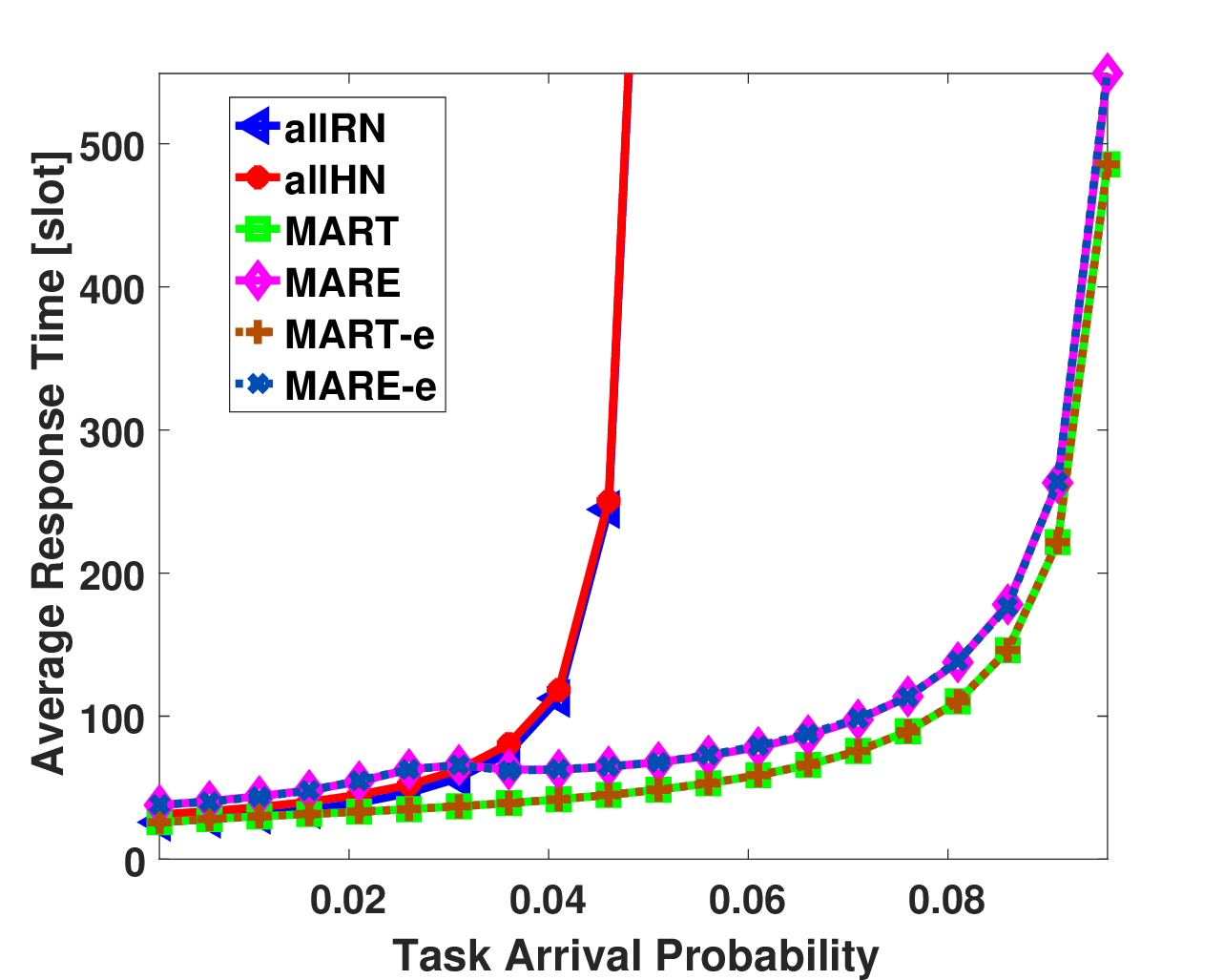}
    \caption{}
    \label{fig:ART_vs_a_both_q_05}
\end{subfigure}%
\begin{subfigure}[t]{0.5\textwidth}
    \centering
		\includegraphics[width=\textwidth]{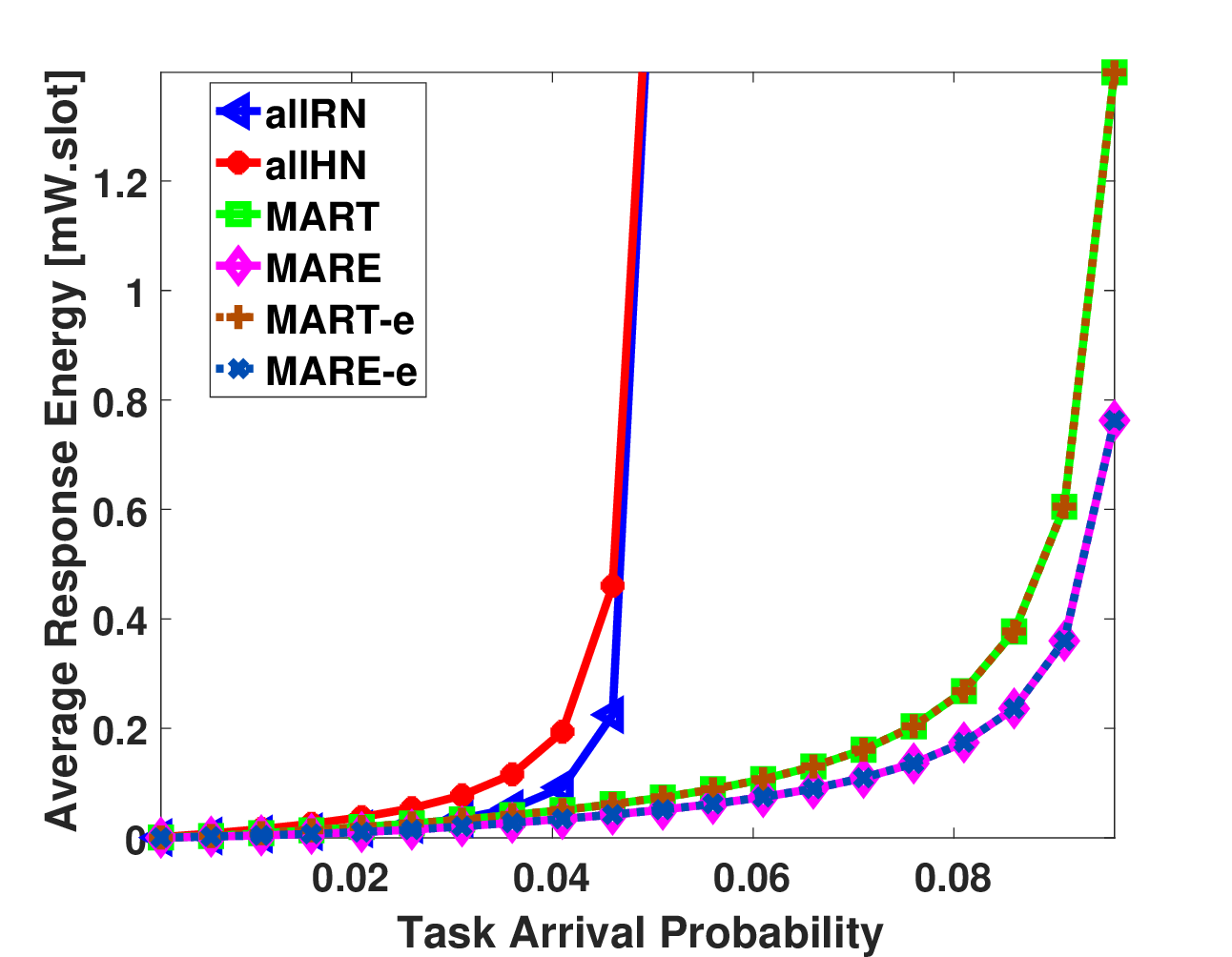}
     \caption{}
    \label{fig:ARE_vs_a_both_q_05}
\end{subfigure}%
 }
\caption{Effect of task arrival probability $a$ on the (a) ART (b) ARE; $q^c_r=0.05, q^c_h=0.05$.} 
 \label{fig:ART_ARE_vs_a_both_q_05}
 \end{figure*}

First, we investigate the validity of the analytic results using Matlab simulations for $10^9$ time slots. Fig.~\ref{fig:ART_ARE_vs_P} displays the ART and the ARE with respect to the transmit power in a system with the task arrival probability $a=0.001$ and different values for $\rho$, $q^c_r$, and $q^c_h$. These settings have been chosen to illustrate different possible patterns of the results, in particular with the MARE scheme. We observe that the simulation results are in accordance with the analytical results. Besides, Fig.~\ref{fig:ART_vs_P} shows that the ART decreases as $P$ increases; the RoD in the ART is large at low values of $P$ but small at high values of $P$. This is in agreement with the analysis presented in the previous section. Also, as expected, the higher the $q^c_r$ and/or $q^c_h$, the lower the ART. Note that at low values of $q^c_r$ and $q^c_h$, the ART in the cases with $\rho=0.33$ and $\rho=0.66$ is considerably lower than that in the cases with $\rho=0$ and $\rho=1$. However, at higher values of $q^c_r$ and $q^c_h$, either there is a slight difference, or the ART in the cases with $\rho=0$ or $\rho=1$ is lower. These indicate that distributing the tasks between the HN and RN helps to reduce the ART, when they have lower probabilities of processing the tasks of the SN; but distributing the tasks may not change or may even increase the ART when the probability of task processing at the RN and/or the HN is high.\looseness=-1

On the other hand, the ARE has different trends. In order to illustrate them clearly and distinguish different charts, we have shown the results for ARE in Fig.~\ref{fig:ARE_vs_P1} and Fig.~\ref{fig:ARE_vs_P2}. It is observed that as $P$ increases, the ARE may increase, decrease or it may decrease first and then increase. This is because, depending on the system settings and due to the effect of several buffers, the RoI in the SAP and RoD in the ART may be different as $P$ increases. Specifically, at low values of $q^c_r$ and $q^c_h$ shown in Fig.~\ref{fig:ARE_vs_P1}, the queues in the computation buffers are large and the ART is so high that increasing the transmit power does not lead to as much RoD in the ART as the RoI in the SAP; therefore, the ARE is minimum at the lowest $P$. At moderate values of $q^c_r$ and $q^c_h$ shown in Fig.~\ref{fig:ARE_vs_P2}, the size of the queues in the computation and transmission buffers are close to each other; hence, as $P$ increases, the RoI in the SAP is comparable with the RoD in the ART and consequently, the minimum ARE happens at a $P$ between the minimum and maximum. At high values of $q^c_r$ and $q^c_h$, the queues of tasks in the computation buffers are very small and the transmission queues have more influence on the ART; as $P$ increases, the RoD in the ART is more than the RoI in the SAP and therefore, the ARE is minimum at the highest $P$.

\begin{figure*}[!t]
  \centering{
\begin{subfigure}[t]{0.5\textwidth}
    \centering
    \includegraphics[width=\textwidth]{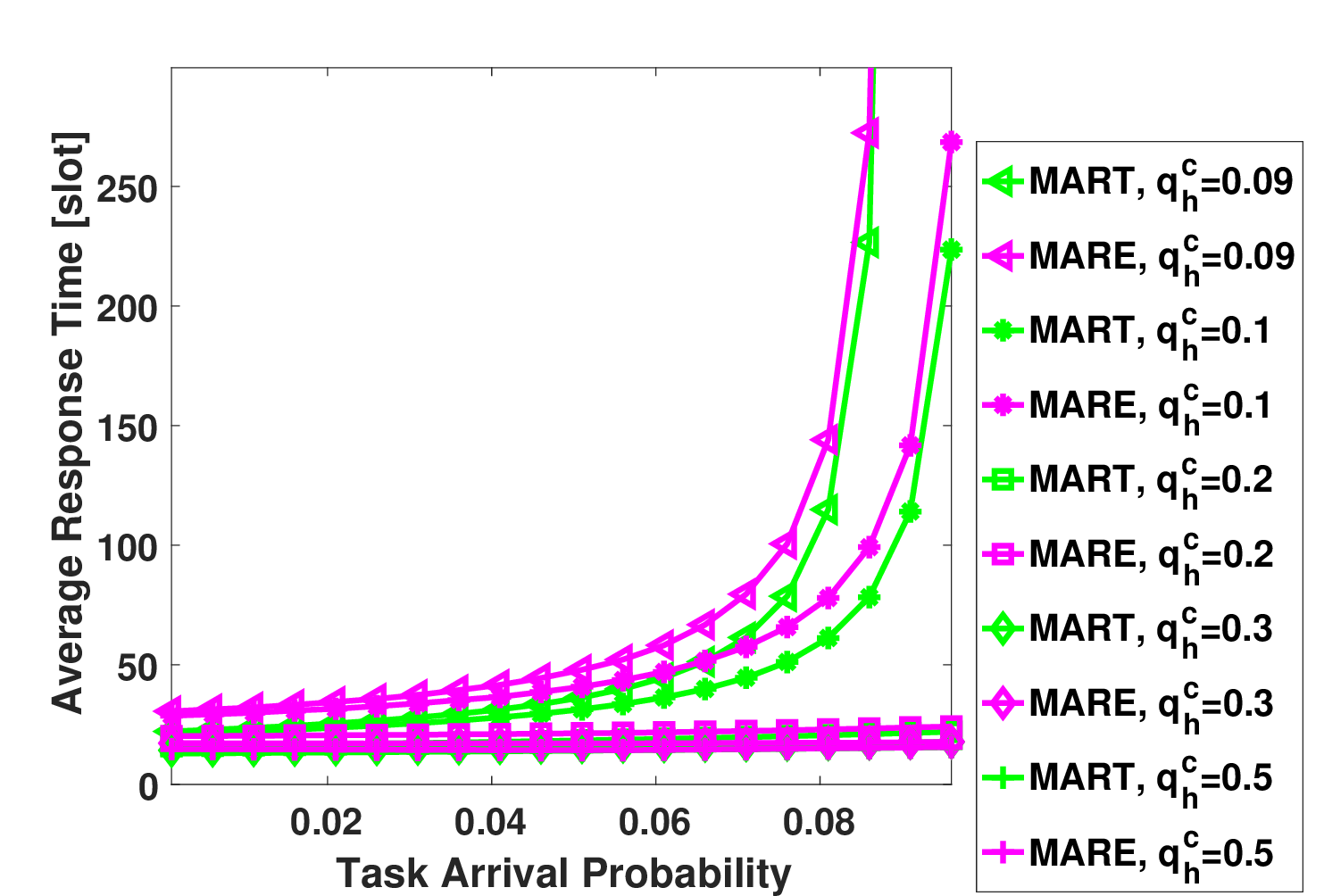}
    \caption{}
    \label{fig:ART_vs_a_qr_001}
\end{subfigure}%
\begin{subfigure}[t]{0.38\textwidth}
    \centering
    \includegraphics[width=\textwidth]{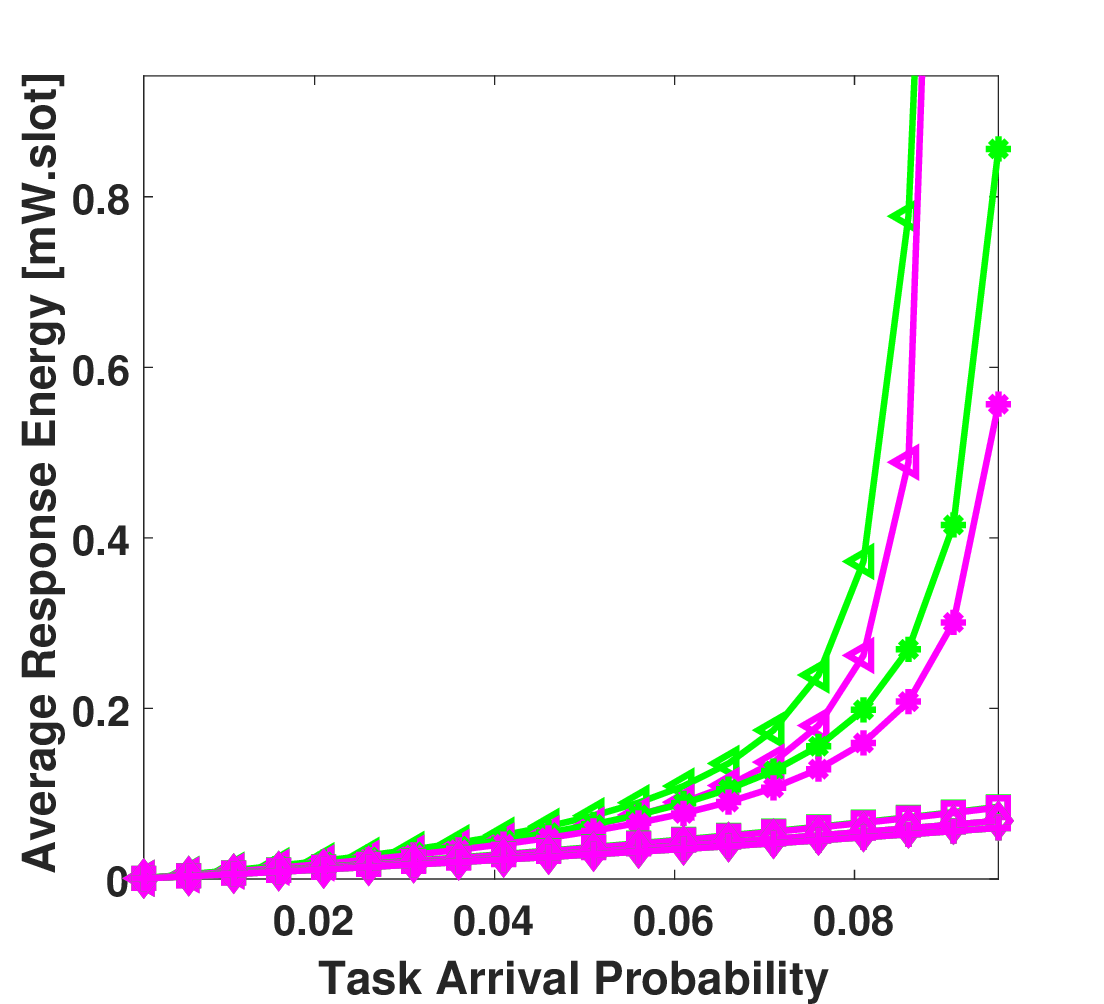}
     \caption{}
    \label{fig:ARE_vs_a_qr_001}
\end{subfigure}\\
\begin{subfigure}[t]{0.33\textwidth}
    \centering
    \includegraphics[width=\textwidth]{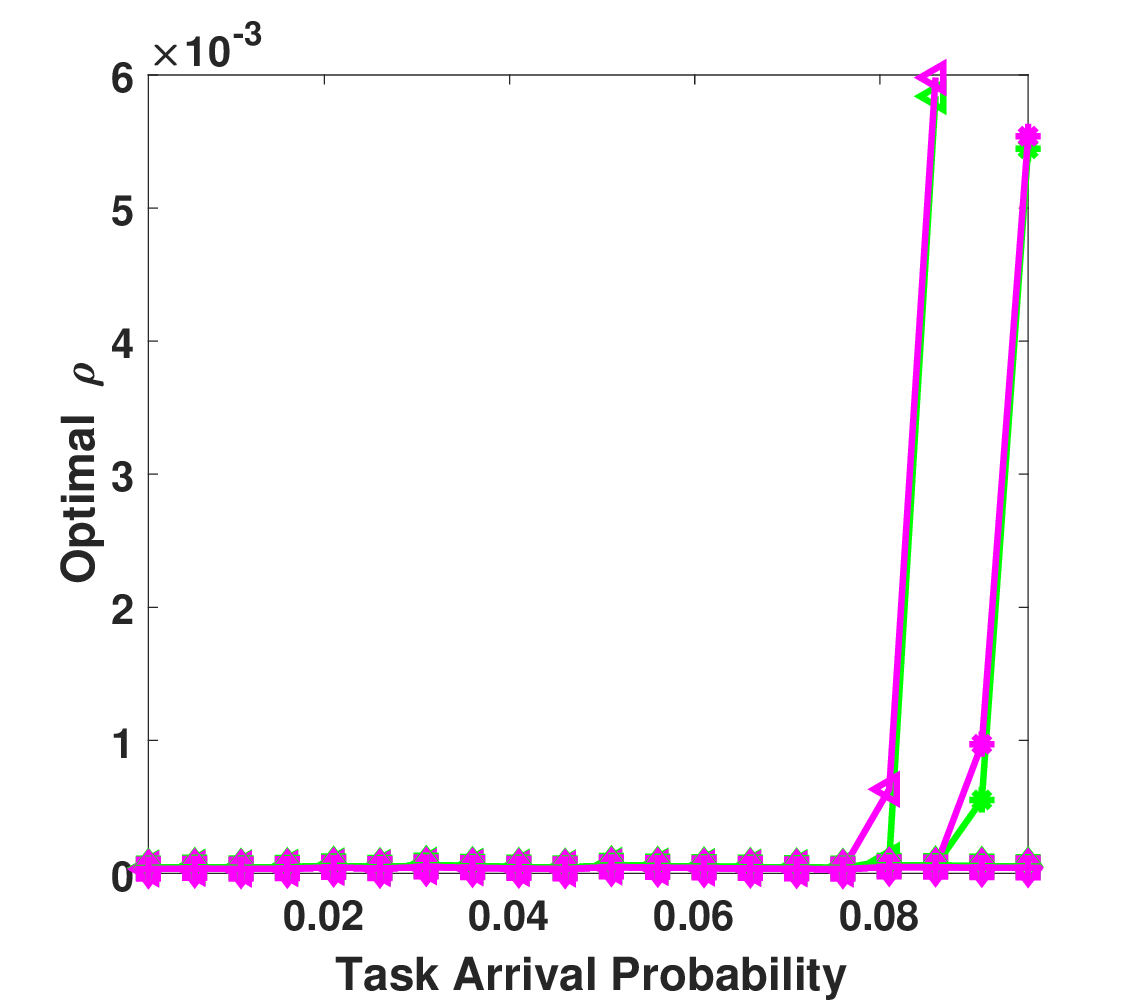}
     \caption{}
    \label{fig:rho_vs_a_qr_001}
\end{subfigure}%
\begin{subfigure}[t]{.33\textwidth}
    \centering
    \includegraphics[width=\textwidth]{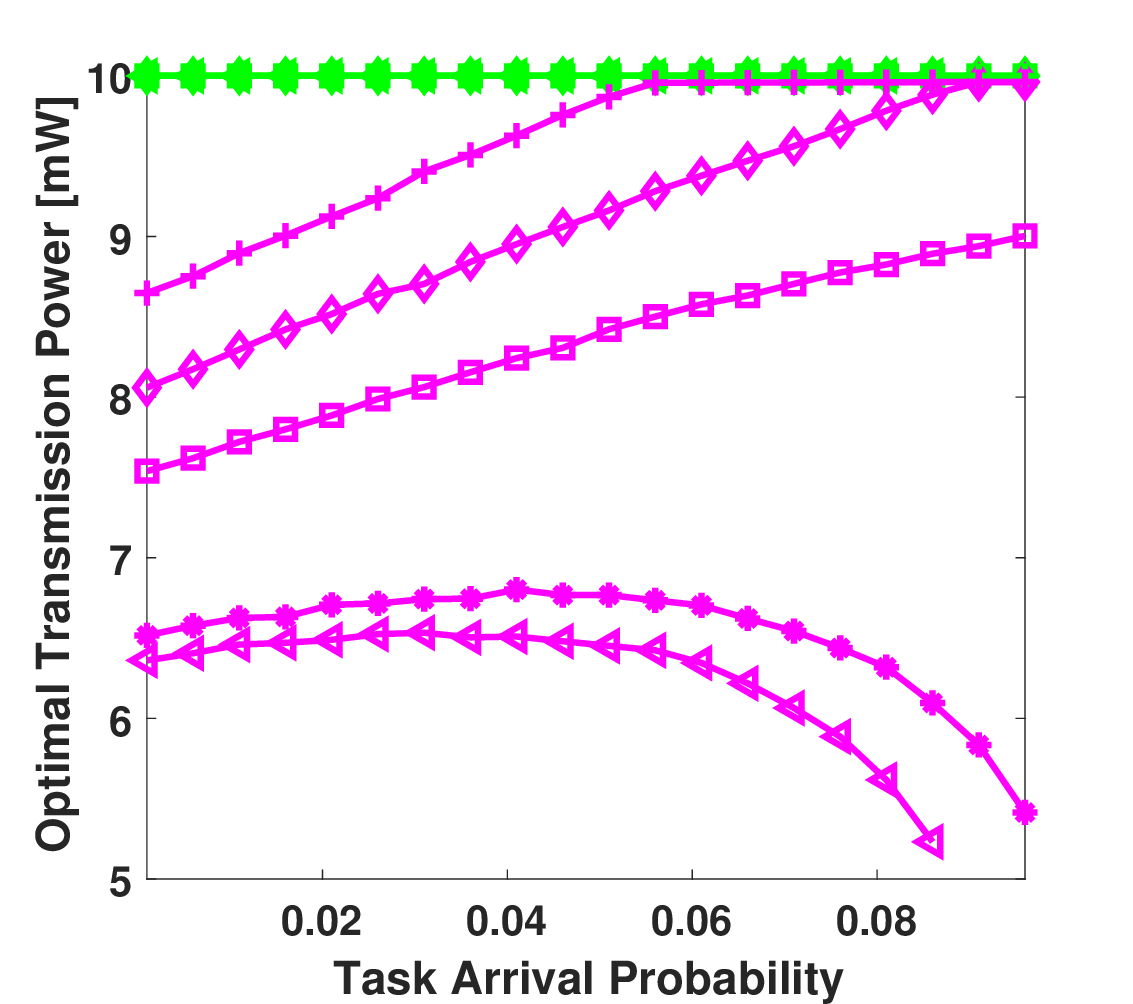}
     \caption{}
    \label{fig:P_vs_a_qr_001}
\end{subfigure}%
\begin{subfigure}[t]{.33\textwidth}
    \centering
    \includegraphics[width=\textwidth]{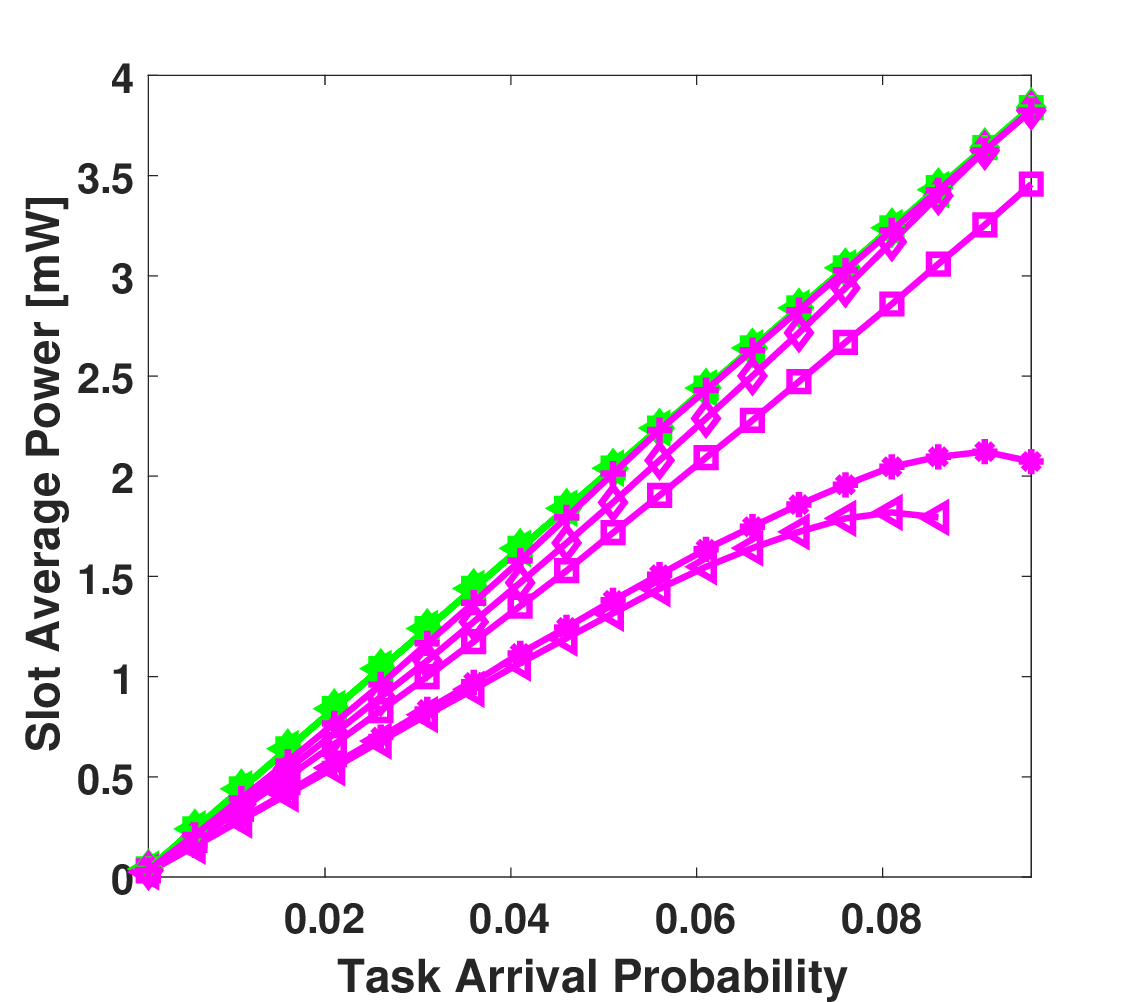}
     \caption{}
    \label{fig:SAP_vs_a_qr_001}
\end{subfigure}%
 }
\caption{Effect of task arrival probability $a$ on the (a) ART (b) ARE (c) $\rho^*$ (d) $P^*$ (e) SAP; $q^c_r=0.001$.} 
 \label{fig:all_vs_a_qr_001}
 \end{figure*}

Next, we investigate the performance of the proposed schemes in comparison with the allRN and allHN baseline schemes, which refer to the methods where all the tasks are assigned respectively to the RN and to the HN. In both of allRN and allHN, the maximum power is used for transmissions. 
Note that since our paper is an early work on buffer-and-server-aided relay-assisted MEC, there are not other similar methods to be used for comparison. Moreover, the allRN and allHN schemes are based on existing methods and are indeed more suitable to show the effect of the proposed schemes because allRN resembles a single-hop time-slotted MEC system and allHN corresponds to the case of a relay-enhanced MEC architecture where the relay station is equipped with transmission buffers but does not have computation capability. For the proposed algorithms, we have considered $\delta=0.0001$ and $K=1$ and have also compared them with the exhaustive search methods MART-e and MARE-e. In the MART-e and MARE-e methods respectively, the MART and MARE problems are solved by exhaustive search over the domains of $P$ and $\rho$ in steps of $\delta$. The task arrival probability $a$ is considered from $0.001$ to $0.096$ in steps of $0.005$ and the maximum transmit power of the nodes is $10$ mW. Note that if the maximum power is used for transmissions, channel service probabilities will be approximately equal to 0.368, which is high above the range of $a$ and will prevent large queuing at the transmission buffers. Hence, in order to guarantee the system stability with the allRN and allHN schemes at least at the beginning (low values of $a$), we consider both $q^c_r$ and $q^c_h$ equal to $0.05$. As shown in Fig.~\ref{fig:ART_ARE_vs_a_both_q_05}, at all values of $a$, the MART scheme has the lowest ART whereas the MARE scheme has the lowest ARE. This also holds in the results presented later and indicates that the proposed methods are highly effective in reaching their goals to minimize the ART or the ARE. The allRN and allHN schemes have a small ART and ARE when $a$ is lower than $0.02$. However, as $a$ increases over $0.02$, the ART and ARE with allRN and allHN increase rapidly and then go towards infinity at around $a=0.045$, whereas this happens for MART and MARE at around $a=0.096$. This indicates that both MART and MARE optimize the use of transmit power and the processing capabilities at the RN and HN, and result in considerably better performance compared with allRN and allHN. Moreover, it is observed that, even with $K=1$, the proposed algorithms achieve almost the same performance as the exhaustive search methods. This has been tested for the scenarios in the sequel, too, and confirms the discussions presented in Remark 1 in the previous section
\footnote{In fact, we conducted many other simulations where, for higher reliability, $K$ was set to values much larger than 1 but the results did not change.}. 
Hence, to have clear figures in the following, we do not illustrate the results of exhaustive search, and only provide the results obtained by the proposed algorithms and focus on the possible trends.

\begin{figure*}[!t]
  \centering{
\begin{subfigure}[t]{0.5\textwidth}
    \centering
    \includegraphics[width=\textwidth]{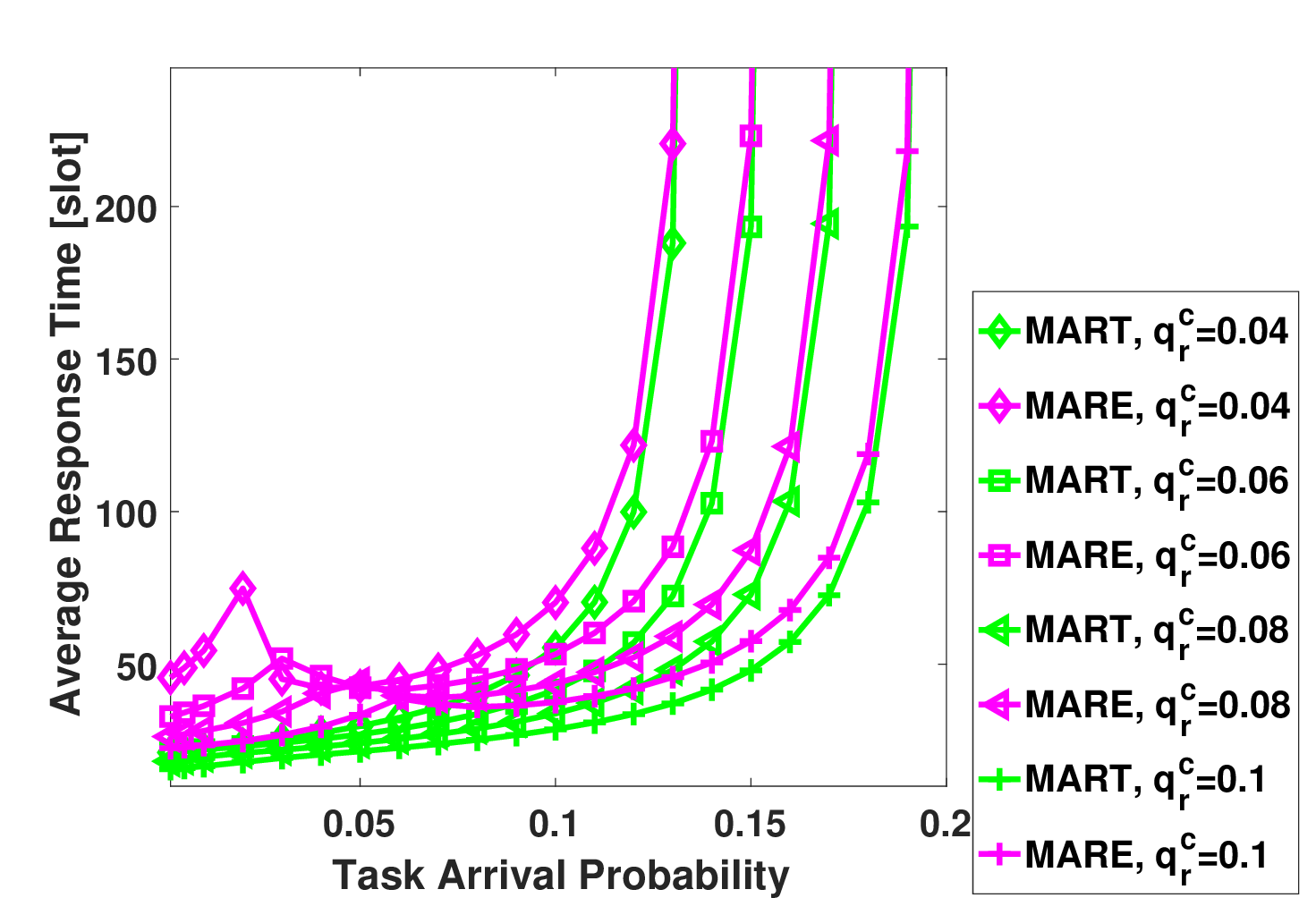}
    \caption{}
    \label{fig:ART_vs_a_qh_1}
\end{subfigure}%
\begin{subfigure}[t]{0.38\textwidth}
    \centering
    \includegraphics[width=\textwidth]{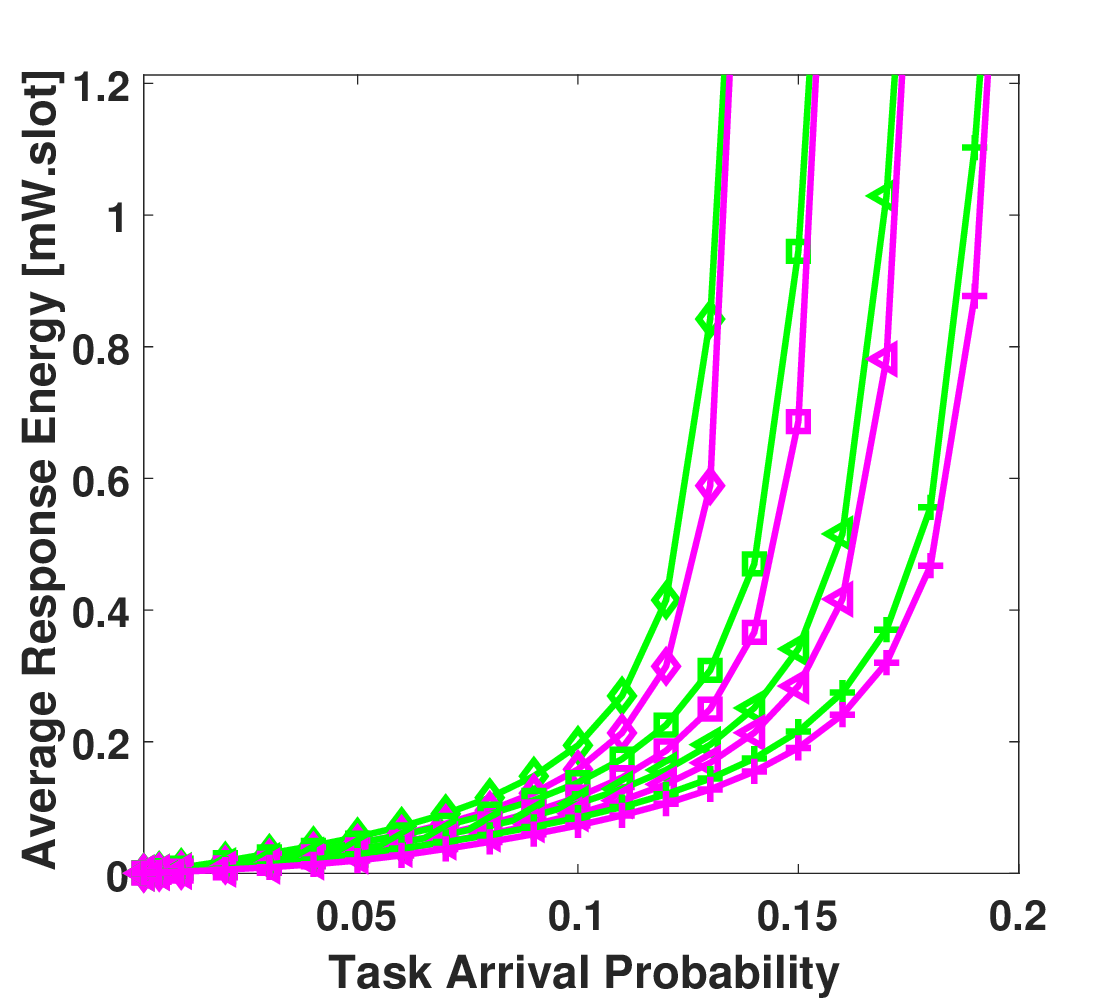}
     \caption{}
    \label{fig:ARE_vs_a_qh_1}
\end{subfigure}\\
\begin{subfigure}[t]{0.33\textwidth}
    \centering
    \includegraphics[width=\textwidth]{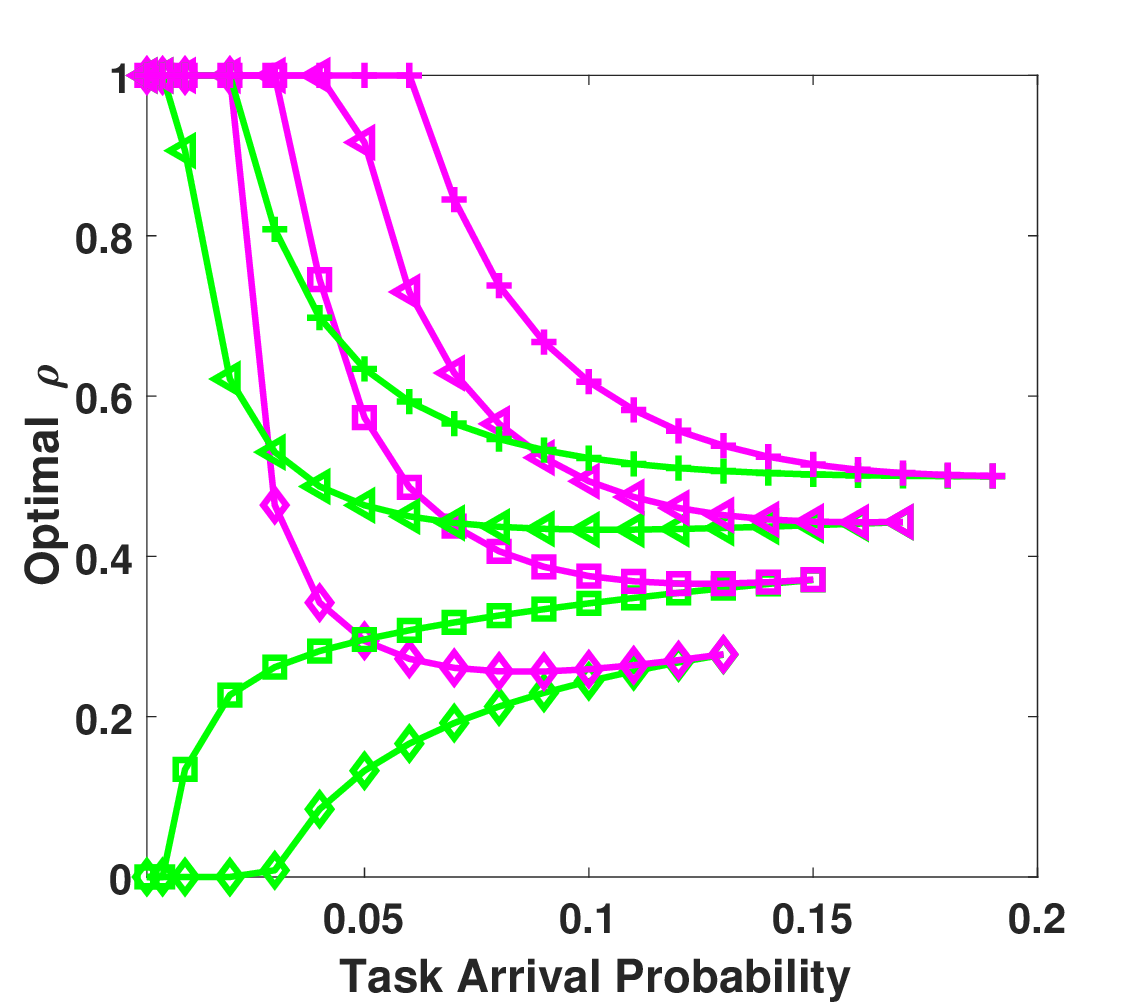}
     \caption{}
    \label{fig:rho_vs_a_qh_1}
\end{subfigure}%
\begin{subfigure}[t]{0.33\textwidth}
    \centering
    \includegraphics[width=\textwidth]{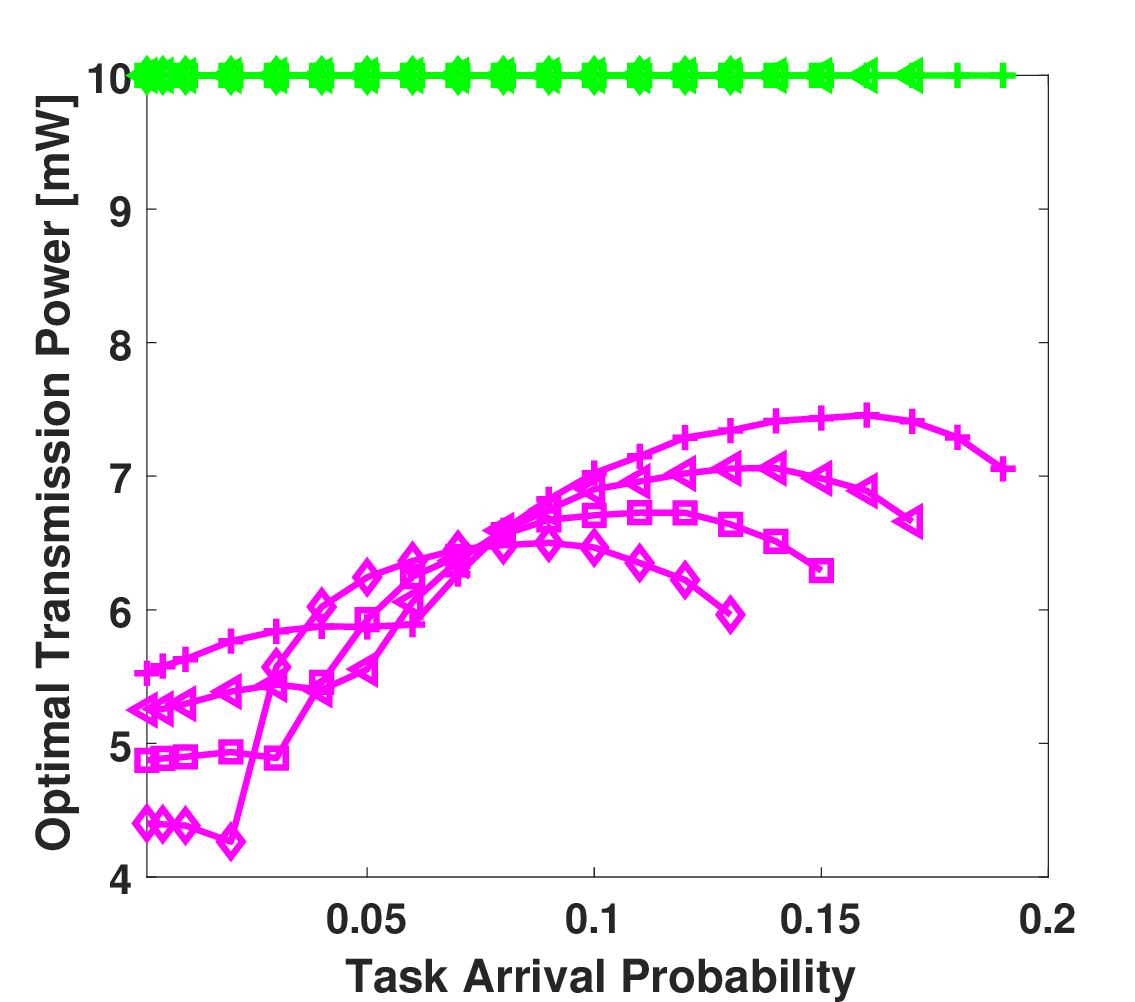}
     \caption{}
    \label{fig:P_vs_a_qh_1}
\end{subfigure}%
\begin{subfigure}[t]{0.33\textwidth}
    \centering
    \includegraphics[width=\textwidth]{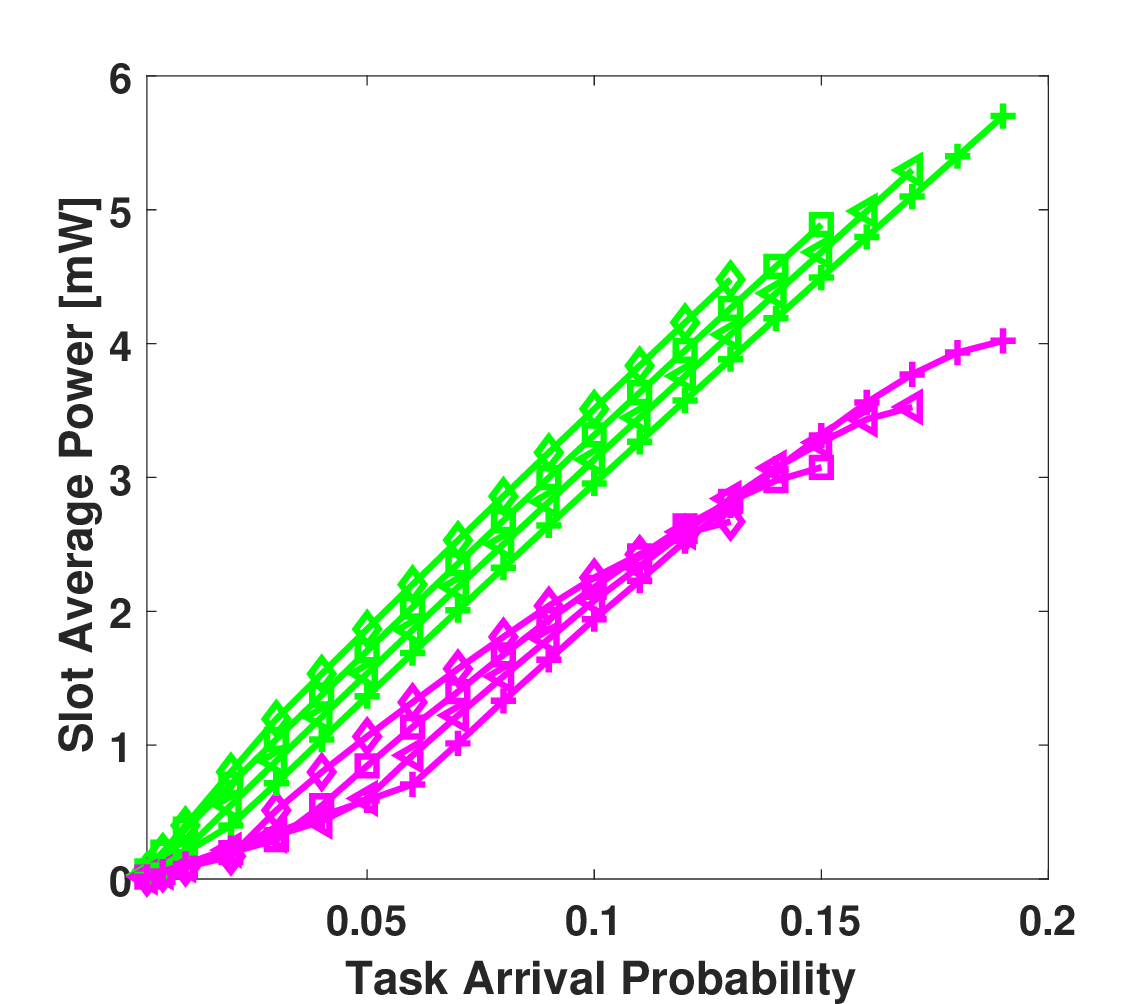}
     \caption{}
    \label{fig:SAP_vs_a_qh_1}
\end{subfigure}%
 }
\caption{Effect of task arrival probability $a$ on the (a) ART (b) ARE (c) $\rho^*$ (d) $P^*$ (e) SAP; $q^c_h=0.1$.} 
 \label{fig:all_vs_a_qh_1}
 \end{figure*}

In order to investigate the performance of the proposed methods in various system settings and provide insights on different trends, we first consider a case with low processing ability at the RN. We fix $q^c_r$ at $0.001$ and examine the impact of the arrival probability of tasks in the scenarios with different values of $q^c_h$. Fig.~\ref{fig:all_vs_a_qr_001} shows the results, where due to the space shortage, all the sub-figures share the same legend. It is clear in Figs.~\ref{fig:ART_vs_a_qr_001} and~\ref{fig:ARE_vs_a_qr_001} that as $a$ goes over $0.06$, the ART and ARE increase rapidly in the cases of low $q^c_h$, i.e., $0.09$ and $0.1$. Particularly, in the case of $q^c_h=0.09$, the MART and MARE problems are not feasible when $a$ goes beyond $0.9$, which is why the results are not illustrated at $a=0.91$ and $a=0.96$. We observe in Fig.~\ref{fig:rho_vs_a_qr_001} that in almost all the considered scenarios, $\rho^*$ is $0$, which means all the tasks are offloaded to the HN; it is only at high $a$ and lower values of $q^c_h$ that $\rho^*$ is a little higher than $0$. This trend is due to the fact that when the computation capability is low at the RN, using its computation server will cause large queuing at the RN computation buffer and consequently, high ART in the system. On the other hand, the computation capability is high at the HN and offloading all the tasks to it leads to a lower ART, except for the cases with high $a$ and lower values of $q^c_h$ which necessitate having few tasks processed at the RN. Fig.~\ref{fig:P_vs_a_qr_001} shows that while with the MART scheme, the maximum power is used for transmissions at the nodes, the MARE scheme achieves comparable ART using different power values. Specifically, at lower values of $q^c_h$, when $a$ increases from zero to relatively moderate values, the MARE scheme increases $P^*$ as its impact is considerable on decreasing the ART. However, when $a$ increases more, increasing the transmit power does not help much in decreasing the ART; thus, MARE reduces $P^*$ to save energy, which is reflected on the SAP in Fig.~\ref{fig:SAP_vs_a_qr_001}. On the other hand, at high values of $q^c_h$, as $a$ increases, MARE uses higher transmit power (even up to the maximum, if needed) because its impact on the RoD in the ART is considerable compared with the RoI in the SAP. This is why we also observe that the more the $q^c_h$, the higher the $P^*$. Fig.~\ref{fig:SAP_vs_a_qr_001} illustrates that the SAP of both MART and MARE increases with task load, which is expected because the more the $a$, the more the task transmissions on the links and the more the power consumption. Moreover, based on the stated discussions earlier, lower SAP is expected for MARE which is confirmed by the results. Since the energy consumption in the system is directly proportional to the SAP and the number of slots the system runs, it can be inferred that as time passes, the MARE scheme can result in remarkable energy saving in the system.

\begin{figure*}[!t]
  \centering{
\begin{subfigure}[t]{0.5\textwidth}
    \centering
    \includegraphics[width=\textwidth]{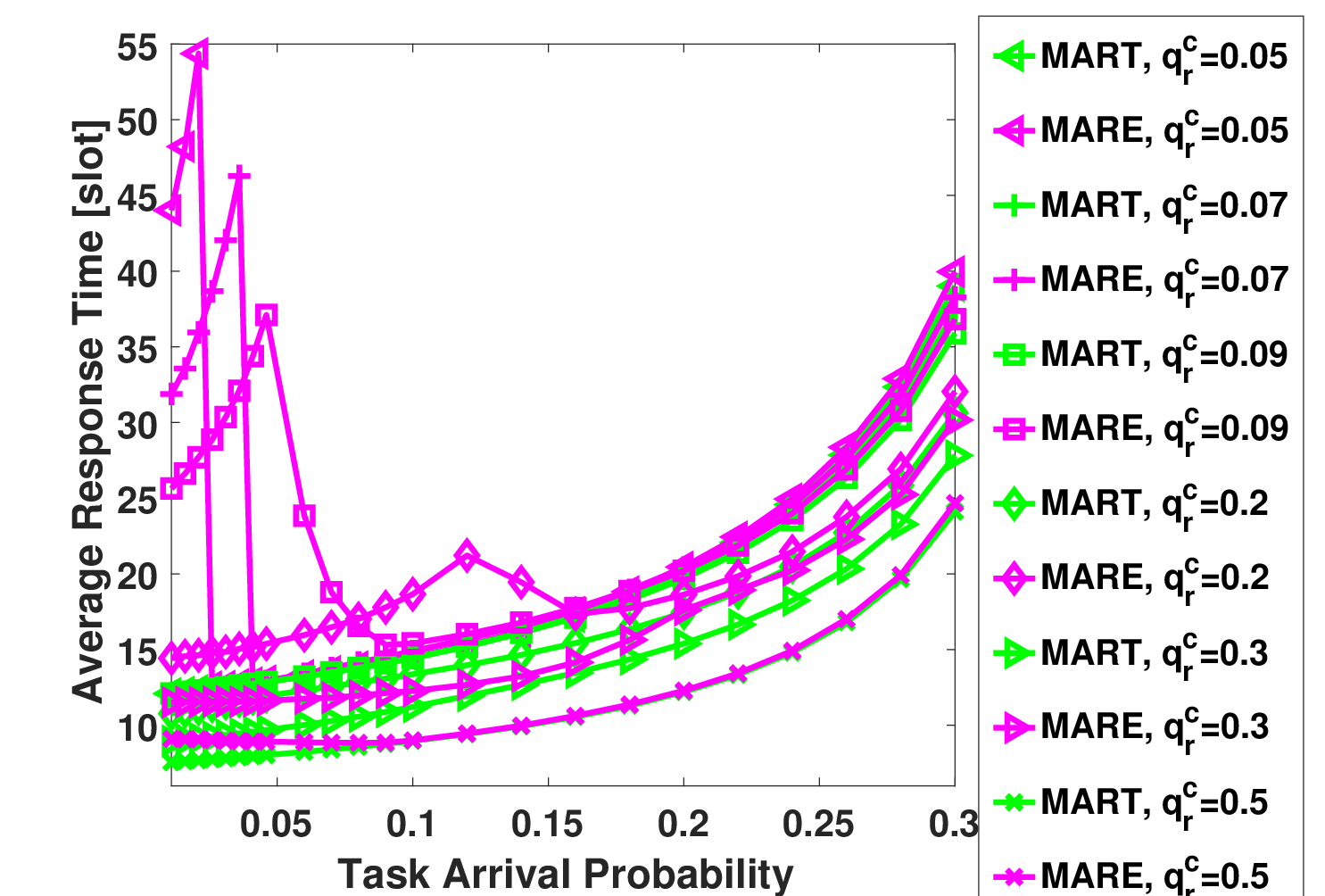}
    \caption{}
    \label{fig:ART_vs_a_qh_99}
\end{subfigure}%
\begin{subfigure}[t]{0.38\textwidth}
    \centering
    \includegraphics[width=\textwidth]{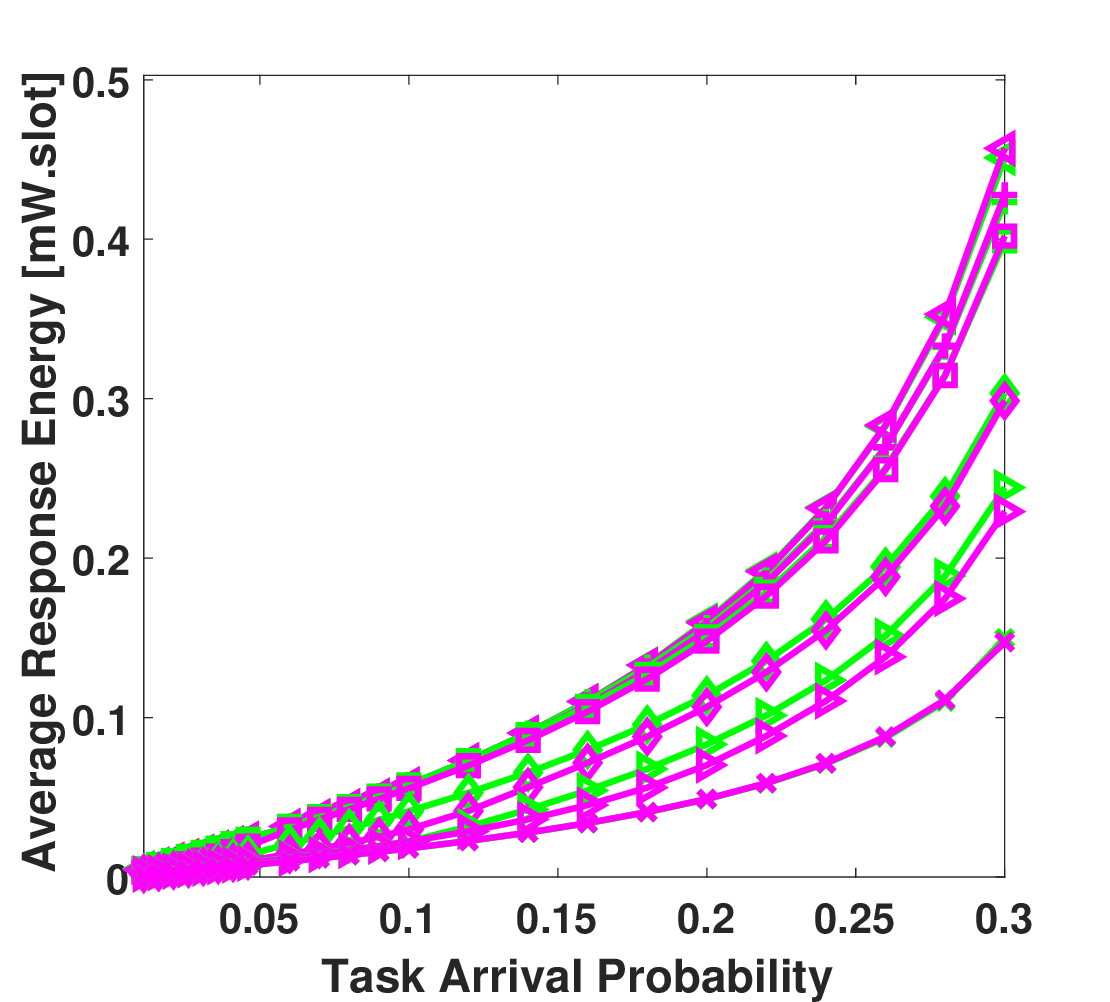}
     \caption{}
    \label{fig:ARE_vs_a_qh_99}
\end{subfigure}\\
\begin{subfigure}[t]{0.33\textwidth}
    \centering
    \includegraphics[width=\textwidth]{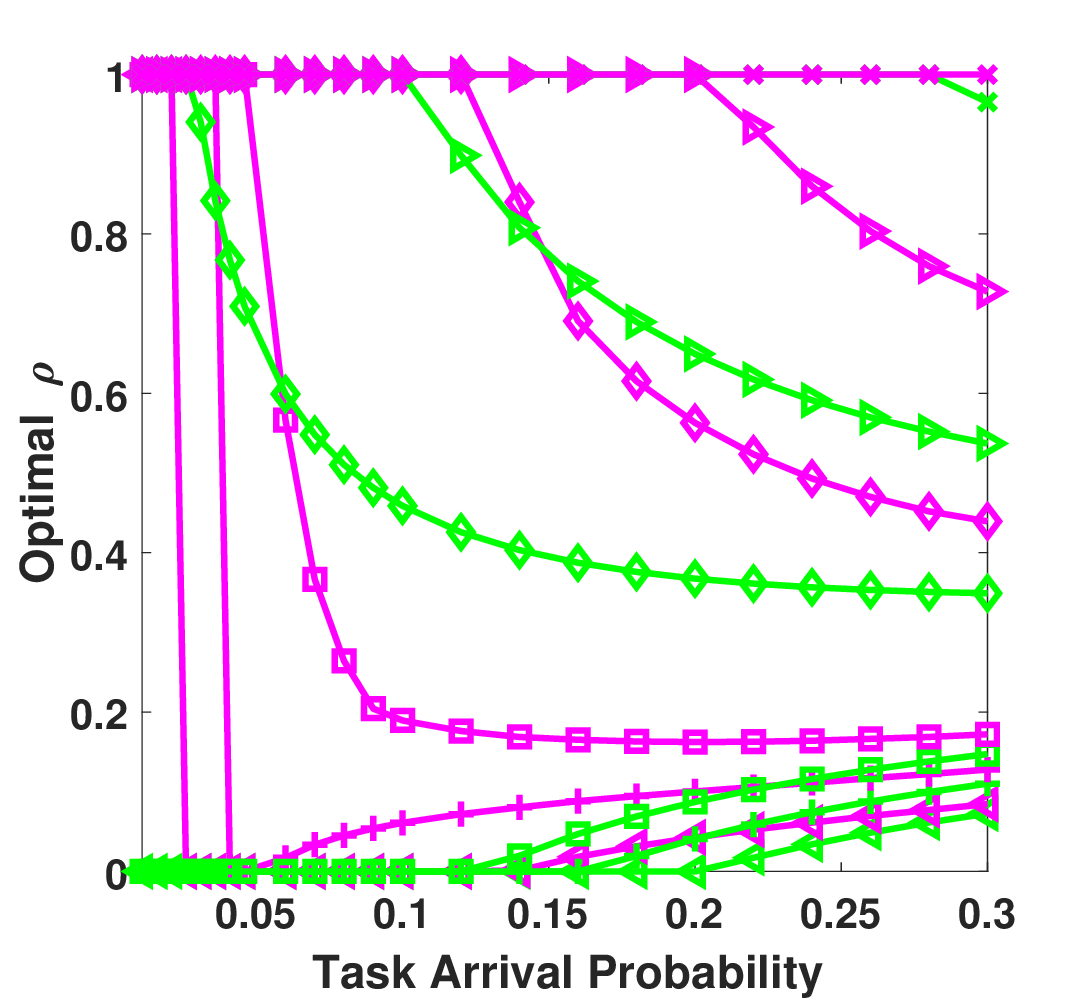}
     \caption{}
    \label{fig:rho_vs_a_qh_99}
\end{subfigure}%
\begin{subfigure}[t]{0.33\textwidth}
    \centering
    \includegraphics[width=\textwidth]{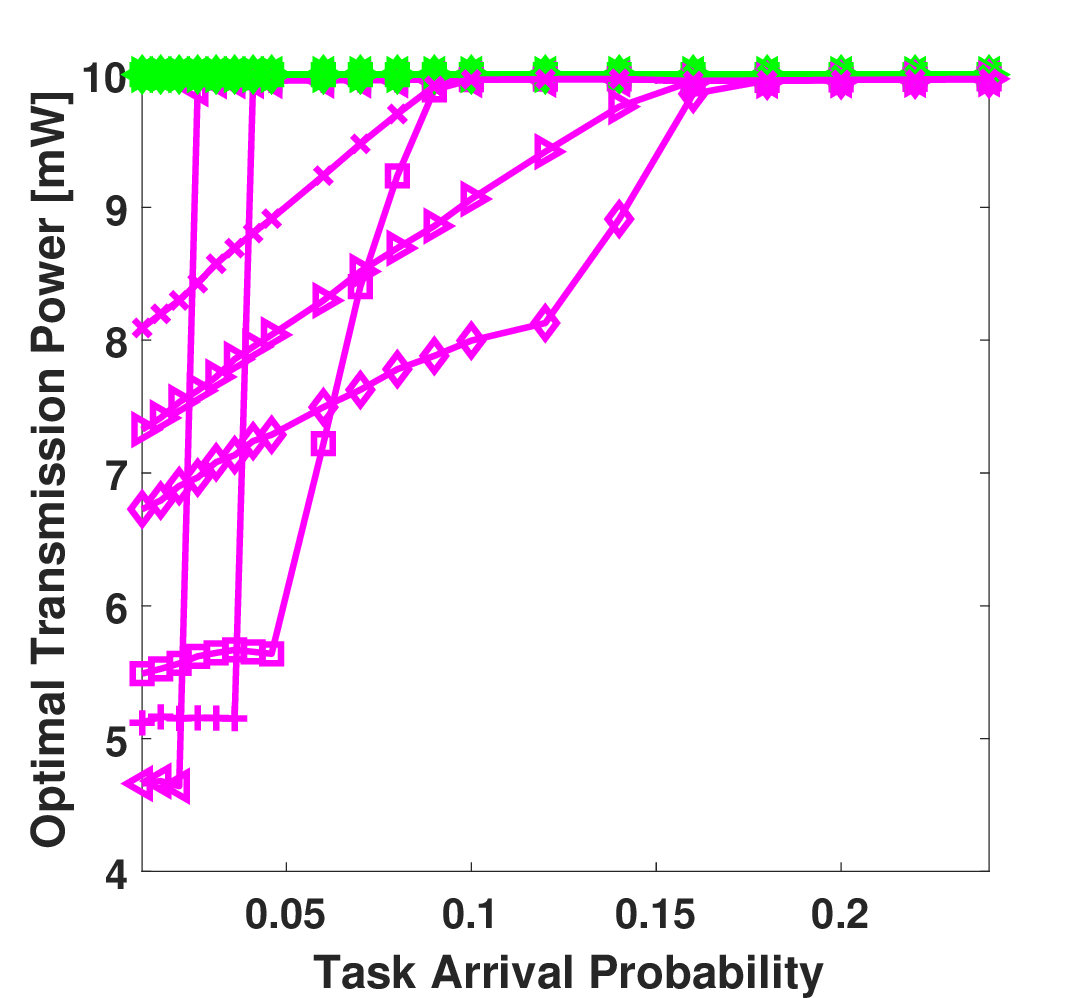}
     \caption{}
    \label{fig:P_vs_a_qh_99}
\end{subfigure}%
\begin{subfigure}[t]{0.33\textwidth}
    \centering
    \includegraphics[width=\textwidth]{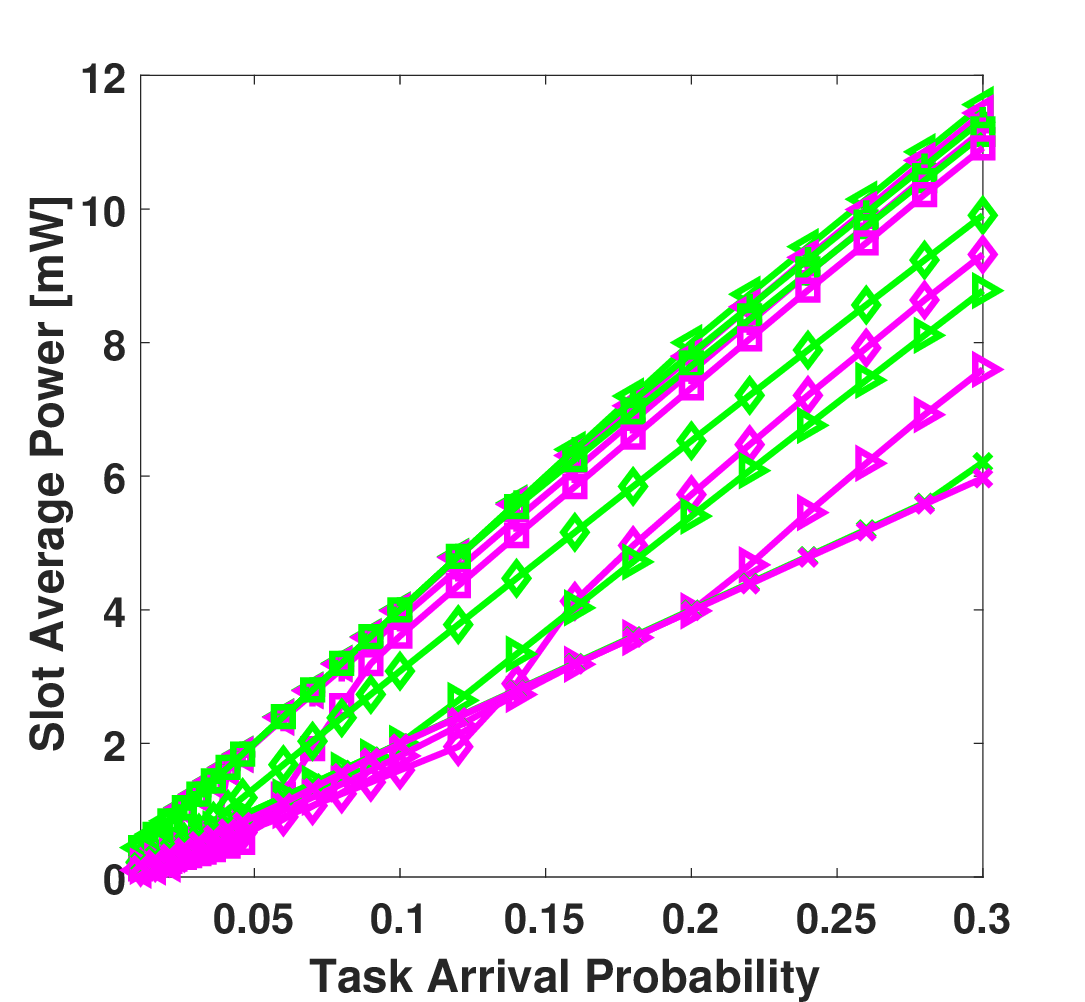}
     \caption{}
    \label{fig:SAP_vs_a_qh_99}
\end{subfigure}%
 }
\caption{Effect of task arrival probability $a$ on the (a) ART (b) ARE (c) $\rho^*$ (d) $P^*$ (e) SAP; $q^c_h=0.99$.} 
 \label{fig:all_vs_a_qh_99}
 \end{figure*}

Next, we fix $q^c_h$ at a relatively high value, i.e., $0.1$, and consider different values for $q^c_r$, from moderate to relatively high, and investigate the effect of increasing the task load. It is observed in Figs.~\ref{fig:ART_vs_a_qh_1} and~\ref{fig:ARE_vs_a_qh_1} that at higher values of $a$, as it increases, the ART and ARE of both the MART and MARE schemes increase continuously towards infinity after some point which is the sign of infeasibility of the MART and MARE problems and instability of the system queues at high task load. The more the $q^c_r$, the larger the range of $a$ at which the system queues are stable; this can also be noticed from Figs.~\ref{fig:rho_vs_a_qh_1},~\ref{fig:P_vs_a_qh_1}, and~\ref{fig:SAP_vs_a_qh_1}. Figs.~\ref{fig:ART_vs_a_qh_1} and~\ref{fig:rho_vs_a_qh_1} indicate that at low task loads, MART achieves the minimum ART by computing all the tasks at the HN ($\rho^*=0$) or the RN ($\rho^*=1$) respectively in the case of moderate or high $q^c_r$. However, as the load increases, $\rho^*$ changes and the tasks are distributed between the RN and the HN, and the more the $q^c_r$, the more the $\rho^*$. It is also observed in Figs.~\ref{fig:P_vs_a_qh_1} and~\ref{fig:SAP_vs_a_qh_1} that even though MART uses maximum power at the nodes, the SAP is different for different values of $q^c_r$. This is due to the fact that the lower the $q^c_r$, the higher the portion of tasks offloaded to the HN and therefore, the higher the average power consumption in each slot. Also note that when $\rho^*=1$, task/result transmissions happen only between SN and RN which is why the SAP has a lower slope in the corresponding cases. 
Fig.~\ref{fig:ART_vs_a_qh_1} shows that the interesting things happen with the MARE scheme at low values of $a$; as $a$ increases, the ART increases up to some point and then returns and decreases. The lower the $q^c_r$, the lower the $a$ and the higher the ART at the return point. These return points correspond to the points in Fig.~\ref{fig:rho_vs_a_qh_1} where $\rho^*$ changes from $1$ to a lower value and the points in Fig.~\ref{fig:P_vs_a_qh_1} where $P^*$ has rapid increase. Their affect is also clear on the SAP in Fig.~\ref{fig:SAP_vs_a_qh_1}. This trend is because the RN has a moderate computation capability and MARE tries to assign all the tasks to the RN in order to not consume power for RN-HN/HN-RN transmissions of the tasks/results. In fact, at lower loads, MARE does not find the RoD in the ART worth the RoI in the SAP. But, that does not hold as the load increases and therefore, it becomes necessary to offload some tasks to the HN. Now, at this situation, the lower the $q^c_r$, the lower the $\rho^*$, and MARE finds the RoI in the SAP worth the RoD in the ART and hence, the transmit power at the nodes are increased. However, as $a$ increases more and the ART increases rapidly, MARE decreases the transmit power, similar to the trends in Fig.~\ref{fig:P_vs_a_qr_001} discussed previously.

In the previous figures, we observed that as the task load increased, $\rho^*$ was either not changed or it changed smoothly between $0$ and $1$. Similarly, the changes in $P^*$ was smooth. It is interesting to note that this might not be the case in some system settings. To show that, we consider scenarios with very high computation capability at the HN, i.e., $q^c_h=0.99$ and moderate to high capabilities for the RN, i.e., $q^c_r$ from $0.05$ to $0.5$. Since the maximum possible channel service probabilities (i.e., at the maximum power) are approximately $0.368$, the ART and the ARE will increase rapidly if the value of $a$ gets close to that. Therefore, we consider $a$ up to $0.3$ to be able to show the details at the lower values of $a$ clearly. Fig.~\ref{fig:all_vs_a_qh_99} shows that the results of MART have trends similar to the ones in the previous figures, but the results of MARE are even more interesting than those shown before. Here, Fig.~\ref{fig:ART_vs_a_qh_99} indicates that in the cases of $q^c_r$ equal to $0.05$ and $0.07$ at low $a$, the ART of MARE has sharp decrease after increase. This can be explained by the sharp change of $\rho^*$ from $1$ to $0$ in Fig.~\ref{fig:rho_vs_a_qh_99} and the sharp change of $P^*$ from around $5$ mW to almost $10$ mW in Fig.~\ref{fig:P_vs_a_qh_99}, as follows. At low $a$, MARE does not find the RoD in the ART worth the RoI in the SAP and achieves lower ARE by assigning all the tasks to the RN. However, when $a$ increases over a specific value, MARE needs to reduce the load of the RN computing server but it finds that in the new situation, offloading all the tasks to the HN leads to a lower ARE compared with the case of distributing the tasks between the RN and HN. Since the delay at the computation buffer of the HN is very low, the delays at the HN-RN and RN-HN transmission buffers are the deciding factor; hence, a transmit power close to the maximum is used at the nodes to decrease the ART and keep the ARE low. Then, as $a$ increases more, MARE needs to reduce the load of the RN-HN and HN-RN transmission buffers and hence, some tasks are assigned to the RN again. Note that in the case of $q^c_r=0.5$ with MARE, $\rho^*$ is $1$ for all the values of $a$. In our simulations, we have observed that this remains so and MARE does not distribute the tasks between the RN and HN, even if we increase $a$ over $0.3$ and up to the points that the ART gets very high and goes towards infinity. This is due to the fact that the maximum service probability of the transmission buffers is $0.368$, and the SN-RN and RN-SN links are the bottlenecks and will cause high ART if $a$ goes over $0.3$. Moreover, even though the HN computation probability is $0.99$, power consumption for offloading some tasks to the HN at high $a$ is not worth the decrease in the ART due to the large queuing in the RN-HN and HN-RN transmission queues. Since $q^c_r=0.5$ is large enough to result in small queuing at the RN computation buffer, assigning all the tasks to the RN leads to lower ARE overall.

\section{Conclusion}\label{sec:conclusion}
This paper has investigated MMEC in the time-slotted systems with block fading channels and fixed transmit power at the nodes. We have studied stochastic offloading scheme in a system with an SN, an RN, and an HN, where SN sends its tasks to the RN and the RN decides in a random way to compute the arrived tasks by its own server or to forward them to be processed at the HN. We have provided a framework which considers the impact of all the computation and transmission queues in the system and, exploiting the queuing theory, we have characterized the ART of the system. We have also proposed the concept of the ARE, as the product of the SAP and the ART, which provides a novel viewpoint about energy efficiency. In order to minimize the ART or the ARE, we have proposed the MART and MARE schemes with the corresponding problem formulations and have analyzed their feasible sets and the objective functions. Based on that, we have presented solution methods and have evaluated their performances. Numerical results confirm the validity of the presented analysis and demonstrate that the proposed schemes significantly improve the ART and ARE compared with the related baseline methods. Moreover, through the extensive evaluations, we have shown the possible trends in the outcomes of the proposed schemes in different scenarios.

\section*{Acknowledgments}
This work was supported by the University of Tabriz [grant number 27/1533].

\bibliography{refs}

\end{document}